\documentclass[preprint,12pt]{elsarticle}

\usepackage{enumitem}
\usepackage{graphicx}
\usepackage{comment}

\usepackage{amsfonts}
\usepackage{amsthm}
\usepackage{amsmath}
\usepackage{amssymb}

\newtheorem{theorem}{Theorem}
\newtheorem{lemma}[theorem]{Lemma}
\newtheorem {observation}[theorem]{Observation}
\newtheorem{claim}[theorem]{Claim}
\newtheorem{corollary}[theorem]{Corollary}

\newcommand{\ceiling}[1]{\left\lceil #1 \right\rceil}
\newcommand{\floor}[1]{\left\lfloor #1 \right\rfloor}

\begin{document}

\begin{frontmatter}

\title{  Multitasking Scheduling with Shared Processing
}

\author[inst1]{Bin Fu}
\ead{bin.fu@utrgv.edu}
\affiliation[inst1]{organization = {Department of Computer Science},
addressline = {University of Texas Rio Grande Valley}, city = { Edinburg}, state = {TX}, postcode={78539}, country={USA}}

\author[inst2]{Yumei Huo\corref{cor1}}
\cortext[cor1]{Corresponding author}
\ead{yumei.huo@csi.cuny.edu}
\affiliation[inst2]{organization = {Department of Computer Science},
addressline = {College of Staten Island, CUNY}, city = { Staten Island}, state = {NY}, postcode={10314}, country={USA}}

\author[inst2]{Hairong Zhao}
\ead{hairong@purdue.edu}
\affiliation[inst2]{organization = {Department of Computer Science},
addressline = {Purdue University Northwest}, city = { Hammond}, state = {IN}, postcode={46323}, country={USA}}

\begin{abstract}

Recently, the problem of multitasking scheduling has attracted a lot of attention in the service industries where workers
frequently perform multiple tasks by switching from one task to another. %Although some research has been done on the productivity of multitasking in the literature, the study of multitask scheduling is still very limited.
Hall, Leung and Li (Discrete Applied Mathematics 2016) proposed a shared processing multitasking scheduling model which allows a team to continue to work on the  primary tasks while processing  the routinely scheduled activities as they occur. The processing sharing is achieved by allocating a fraction of the processing capacity to routine jobs and the remaining fraction, which we denote as sharing ratio, to the primary jobs. % , which we call as processor sharing ratio,   will be
%  The authors studied several scheduling problems on a single machine assuming the sharing ratio is a constant.
In this paper, we generalize this model to parallel machines and allow the fraction of the processing capacity assigned to routine jobs to vary from one to another. The objectives are minimizing makespan and minimizing the total completion time.
We show that for both objectives, there is no polynomial time approximation algorithm unless $P=NP$ if the sharing ratios
%, which are the fractions of the processing capacity assigned to primary jobs,
are arbitrary for all machines. Then we consider the problems where the sharing ratios on some machines have a constant lower bound. For each objective, we analyze the performance of the classical scheduling algorithms and their variations and then develop a polynomial time approximation scheme
%that runs in linear time
when the number of machines is a constant.
%For the objective of makespan, we analyze the performance of the popular  scheduling rules and their variations including List Scheduling (LS), Longest Processing Time first (LPT),    LS-ECT and LPT-ECT and then develop one approximation scheme that runs in linear time when the number of machines is a constant. For the objective of total completion time,
%we analyze the performance of SPT and SPT-ECT rule and develop an approximation scheme.

\end{abstract}

\end{frontmatter}

\section{Introduction}\label{intro}
Recently, the problem of multitasking scheduling has attracted a lot of attention in   the service industries where workers frequently perform multiple tasks by switching from one task to another. Initially the term multitasking (or time sharing) was used to describe the sharing of computing processor capacity among a number of distinct jobs (\cite{d71}).
%In the workplace
Today, human often work in the mode of multitasking in many areas such as health care, where 21\% of hospital employees spend their working time on more than one activity \cite{olb06}, and  in consulting where workers usually engage in about 12 working spheres per day \cite{gm05}.
% As for the effect of multitasking, some research  \cite{vmn08} showed that by working on several different tasks over a given interval of time, a multitasker can improve overall productivity by reducing the time spent waiting between jobs, and it is more efficient for the worker to switch to a new task rather than idly waiting on a pending task. Some other  studies (\cite{h05}, \cite{mm03}), however,  have  indicated that multitasking may disrupt work and may result in a significant loss of productivity. Thus, it is very important to investigate the effects of multitasking on the productivity from different perspectives.
Although in the literature some research has been done on the effect of multitasking (\cite{gm05}, \cite{cip14}, \cite{ssmw13}), the study on multitasking in the area of scheduling is still very limited (\cite{hll15}, \cite{hll16}, \cite{sh15}, \cite{zzc17}).

Hall, Leung and Li \cite{hll16} proposed a multitasking scheduling model that allows a team
 to continuously   work on its main, or primary tasks while  a fixed percentage of its processing capacity may be allocated to process the routinely scheduled activities as they occur.
Some examples of the routinely scheduled activities are administrative meetings, maintenance work, or meal breaks. In these scenarios, some team members will perform these routine activities while the remaining team members can still focus on the primary tasks. An application of this model is in the call center where during a two-hour lunch period each day one half of the working team will take a one hour lunch break during each hour so that no customers' calls will be missed. In these examples, a working team can be viewed as a machine which may have some periods during which routine jobs and primary jobs will share the processing.

In many practical situations of this multitasking scheduling model, since the routine activities are essential to the maintenance of the overall system, they are usually managed separately and independent of the primary jobs. Some third-party companies such as Siteware, provide  routine activities management for other companies. They help plan the routine activities for all the teamwork including the release times and duration of the routine jobs, the priority of the routine jobs, and the team members to whom the routine jobs can be assigned, etc. As described in the website of Siteware, %it is of vital importance that all routine processes are standardized for teamwork, categorized by their request of resources and level of collaboration, prioritized based on their deadlines, and performed correctly within scheduled times.
%The company addresses that there is a necessity to delegate
usually routine jobs are assigned to the employees of the teams based on two criteria: professional skills and procedure priority. Professional skills criteria refers to handing over routine duties according to employees’ skills. Procedure priority refers to delegate routine activities that do not demand special attention and can be done by other people without major issues. The service provided by companies like Siteware motivates our model such that when the primary jobs are all available at time 0 for scheduling, the release times and duration of the routine jobs are all predetermined and are known beforehand and the machine capacity for each routine job can be predetermined as well.
%processing primary jobs can be a constant for each shared processing interval where
%As described in the website of siteware.co, work routine management, which is indispensable for companies that intend to improve their performance and achieve better results, usually predetermines and plans the routine jobs for all the teamwork including
%predetermined and planned independent of the primary jobs. For example, the release times and duration of the routine jobs  such as administrative meetings, maintenance work, lunch breaks etc.,  are known beforehand. And in many cases, the routine jobs only need to be or should be  assigned to some team members. Routine jobs management
%Take the company siteware.co as an example. Siteware.co provides the service to other companies by finding them the solutions of work routine management. As described in the website of siteware.co, work routine management is indispensable for companies that intend to improve their performance and achieve better results.
%To exemplify this scenario, one can refer to the systems where the on-time performance of routine jobs is critical.
In \cite{hdv14}, routine jobs are also discussed in that the relative priority given to these routine activities is determined by established rules.

In these practical situations, due to the predetermined routine jobs, a working team is viewed as a machine which may have a sequence of time intervals with different processing capacities that are available for primary jobs.
%the fact that only some team members are delegated to process routine jobs and sometimes the urgency of the primary jobs, the positive but maybe low machine capacity %(the processing ratio e is greater than 0 but a relatively low value) is still given to the primary jobs.
%In these situations, the  machine capacity for
%processing primary jobs can be a constant for each shared processing interval where
%each routine job can be predetermined.
In the scheduling model of \cite{hll16}, it is assumed that the machine capacity is the same for all routine jobs and there is only a single machine.
In this paper, we generalize the model from \cite{hll16} to  parallel machine environment. Moreover, the machine capacity allocated to routine jobs can vary from one to another, instead of being same for all routine jobs as in  \cite{hll16}. This is more practical considering there may exist different routine jobs with different priorities. Hence, the goal is to schedule the primary jobs to the machines subject to the varying capacity constraints so as to minimize the objectives.

In this paper we also consider the case that on some machines the machine capacity allocated to routine jobs is limited and thus the machine capacity allocated to primary jobs may have a constant lower bound. This scenario can be exemplified by the following real life applications. In
%many circumstances, it is necessary to have service continuously available for primary jobs, such as in
many companies' customer service and technical support departments, the service must be continuous for answering customers' calls and for troubleshooting the customers' product failures.
So a minimum number of members from the team are needed to provide these service at any time while the size of a team is typically of ten or fewer members as recommended by Dotdash Meredith Company in their management research. To model this, we allow  the capacities allocated for primary jobs on some machines to have a constant lower bound.

\subsection{Problem Definition}

Formally, our problem can be defined as follows. We are given $m$ identical machines $\{M_1, M_2, \ldots, M_m\}$ and a set $N = \{1, \ldots, n\}$ of primary jobs that are all available for processing at time 0. Each primary job $j \in N $ has a processing time $p_j$ and can be processed by any one of the machines uninterruptedly.
%A primary job can only be processed on one machine, and cannot be preempted.
We are also given a set of routine jobs that have to be processed during certain time intervals on certain machines.
When a routine job is processed, a fraction of the machine capacity is given to the routine job and the remaining machine capacity is still used for the continuation of the primary job. If a primary job shares the processing with the routine job but completes before the routine job,
then the next primary job will be immediately started and share the processing with this routine job. On the other hand, if the routine job shares the processing with a primary job but completes before the primary job and no other routine job is waiting for processing, this primary job will immediately have full capacity of the machine for processing. We use $k_i$ to denote the number of routine jobs that need to be processed on machine $M_i$, and  $\tilde {n}$ to denote the total number of the routine jobs, i.e. $\tilde{n}  = \sum_{1 \le i \le m} k_i$.

\begin{figure}
 \includegraphics[width=\textwidth]{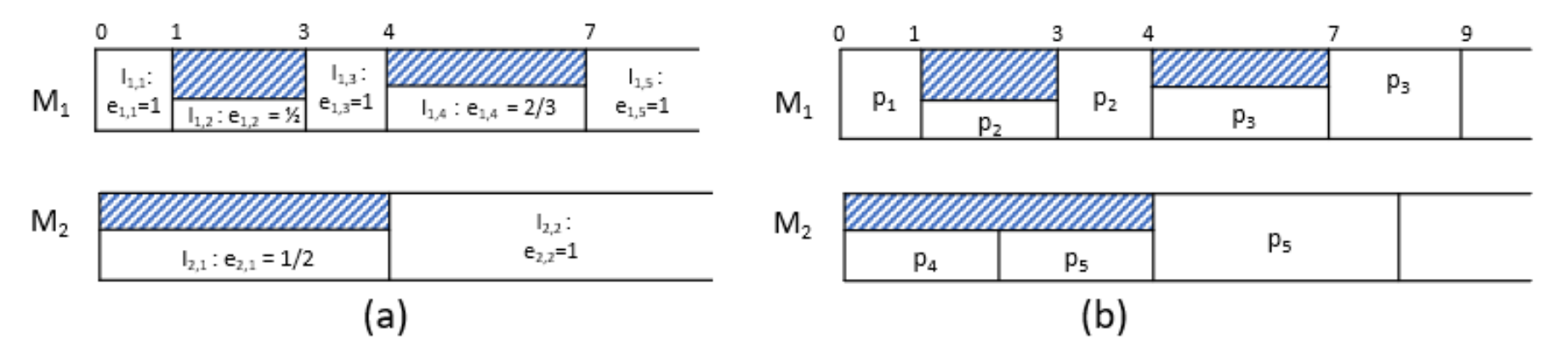}
\caption{An example of multitasking scheduling on 2 machines with shared processing, 3 routine jobs are shown as shaded intervals. (a) The intervals and the sharing ratios; (b) A schedule of the 5 primary jobs with the processing times of 1, 2, 4, 1, 5. } \label{fig:sharing intervals}
\end{figure}
%As we addressed earlier, routine job management is independent from the primary job. It has been predetermined during which time interval a routine job will be processed on which machine with how much machine capacity.
As we addressed earlier, routine job management, independent from the primary job scheduling, has predetermined the time interval during which a routine job is processed, the machine on which the routine job is processed, and the machine capacity that is assigned to a routine job.
%sharing ratio during the interval, and  are all fixed and known beforehand.
So we are only concerned with the schedule of the primary jobs when some fraction of the machine capacity has been assigned to the routine jobs during some time intervals on some machines.
%We are  interested in the intervals where we can schedule the primary jobs and how much capacity is allocated to primary jobs.
In this sense, we can view each machine consisting of intervals with full capacity alternating with those intervals with a fraction of machine capacity, and we use  ``sharing ratio'' to refer to the fraction of the capacity available to the primary jobs, which is in the range of $(0,1]$. Apparently, each machine $M_i$ $(1 \le i \le m)$ has  $O(k_i)$ intervals in total.
%sharing ratio less than 1.
Without loss of generality,
%we assume that the $n$ primary jobs are input in arbitrary order,
we assume that these intervals are given in  sorted order, denoted as  $I_{i,1} = (0, t_{i,1}]$, $I_{i,2} = (t_{i,1}, t_{i,2}]$, $\ldots$, and their corresponding sharing ratios are $e_{i,1}$, $e_{i,2}$, $\ldots$, all of which  are  in the range of $(0,1]$, see Figure~\ref{fig:sharing intervals} (a) for an illustration of machine intervals and Figure~\ref{fig:sharing intervals} (b) for an illustration of one schedule of primary jobs in these intervals.

%For any time point $t$, we let $A_i(t)$ denote the total amount of processing time of the jobs that can be processed during $(0,t]$ on machine $M_i$. Formally, $A_i(t)={t}\cdot{e_{i,1}}$ if $t\in I_{i,1} = (0, t_{i,1}]$, and $A_i(t)=A_i(t_{i, k})+{(t-t_{i,k})}\cdot{e_{i, k+1}}$ if $t\in I_{i, k+1}=(t_{i,k}, t_{i, k+1}]$.
%By the definition of $A_i(t)$, if $S_i$ is a set of jobs assigned to machine $M_i$ and $\sum_{j\in S_i} p_i\le A_i(t)$, then the jobs in $S_i$ can be completed before time $t$ on machine $M_i$.
%Let $A(t)$ be the total amount of processing time of the jobs that can be processed during $(0,t]$ on all machines. That is, $A(t)=\sum_{i=1}^m A_i(t)$.
%that is, if $t_{i,k} < t \le t_{i,k+1}$, $A_i(t) = ( \sum_{u \le k} (t_{i,u}-t_{i,u-1})e_{i,u}) + (t-t_{i,k}) e_{i,k+1}$.

% /////////////////////////////

% %To exemplify this scenario, one can refer to the systems where the on-time performance of routine jobs is critical. Hence,
% %
% %Therefore, we can define our machine environment such that for each machine $i (1\le i\le m)$, a sequence of time intervals $I_{i,k}$ with the processing capacity ratio of $e_{i,k}$ are available for primary jobs' processing.
% ////////////////////////////////

In this paper, we focus on two objectives: minimizing the makespan and minimizing the total completion time of the primary jobs. For any schedule $S$, let $C_j(S)$ be the completion time of the primary job $j$ in $S$. If the context is clear, we use $C_j$ for short. The makespan of the schedule $S$ is $C_{max}(S)=\max_{1\le j \le n} \{C_j\}$ and the total completion time of the schedule $S$ is $\sum C_j(S)=\sum_{1\le j \le n} \{C_j\}$. Using the three-field $\alpha \mid \beta \mid \gamma$ notation introduced by Graham et al. \cite{gllr79}, our problems are denoted as $P_m, e_{i,k} \mid \mid C_{max}$ and $P_m, e_{i,k} \mid \mid \sum C_j$ if the sharing ratios are arbitrary; % for all shared intervals (routine jobs);
if all the sharing ratios have a constant lower bound $e_0$ ($0 < e_0 \le 1$), our problems are denoted as $P_m, e_{i,k} \ge e_0 \mid \mid C_{max}$ and $P_m, e_{i,k} \ge e_0  \mid  \mid \sum C_j$; and if the sharing ratio is at least $e_0$ for the intervals on the first $m_1$ ($ 1 \le m_1 \le m-1$) machines but arbitrary for other $(m-m_1)$ machines, our problems are denoted as $P_m, e_{i \le m_1,k} \ge e_0 \mid \mid C_{max}$ and $P_m, e_{i \le m_1,k} \ge e_0  \mid  \mid \sum C_j$.

\subsection{Literature Review}
The shared processing multitasking model studied in this paper was first proposed by Hall et. al. in \cite{hll16}. They studied this model in the single machine environment and assumed that the sharing   ratio   is a constant $e$ for all the shared intervals. For this model, it is easy to see that the makespan is the same for all schedules that have no unnecessary idle time. The authors in  \cite{hll16}  showed that the total completion time can be minimized by scheduling the jobs in non-decreasing order of the processing time, but  it  is unary NP-Hard for the objective function of total weighted completion time. When the primary jobs have different due dates, the authors    gave  polynomial time algorithms for  maximum lateness and the number of late jobs.
%Then, Sum et al. \cite{slh}, ??? Sum et al. \cite{sh15}, ??? and Zhu et al. \cite{zzc17} ??? extended their work and study the problem such that the processing of a selected task suffers from interruption by other tasks that are available but unfinished.

For the related work, Baker and Nuttle \cite{bn80} studied the problems of sequencing $n$ jobs for processing by a single resource to minimize a function of job completion times subject to the  constraint that the availability of the resource varies over time. The motivation for this machine environment comes from  the situation in which processing requirements are stated in terms of labor-hours and labor availability varies over time. The example can be found in the applications of rotating Saturday shifts, where the company only maintains a fraction, for example 33\%, of the workforce every Saturday. The authors showed that a number of well-known results for classical single-machine problems can be applied with little or no modification to the corresponding variable-resource problems. Hirayama and Kijima \cite{hk92} studied this problem when the machine capacity varies stochastically over time. Adiri and Yehudai \cite{ay87} studied the problem on single and parallel machines such that if a job is being processed, the service rate of a machine remains constant and the service rate can be changed only when the job is completed.

So far there are no results about the problems studied in this paper. Note that  if $e_{i,k} = 1$ for all time intervals, that is, there are no routine jobs, our problems become the classical parallel machine scheduling problems $P_m \mid \mid C_{max}$ and $P_m \mid \mid \sum C_j$. The problem $P_m \mid \mid \sum C_j$ can be solved optimally using SPT (Shortest Processing Time First) rule, which schedules the next shortest job to the earliest available machine. The problem $P_m \mid \mid C_{max}$ is a NP-Hard problem and some approximation algorithms have been designed for it. Graham showed in \cite{g66} that LS (List Scheduling) rule generates a schedule with an  approximation ratio of $2-\tfrac{1}{m}$. The LS rule schedules the jobs one by one in the given ordered list. Each job is assigned in turn to a machine which is available at the earliest time. If the given list is in non-increasing order of the processing time of the jobs, the list schedule rule is called LPT ( the Longest Processing Time First) rule. In \cite{ay87},  Graham showed that LPT  rule generates a schedule whose approximation ratio is $(\tfrac{4}{3} - \tfrac{1}{3m})$.
%assigns at time t = 0 the m longest jobs to the m machines. After that, whenever a machine is freed the longest job among those not yet processed is put on the machine.
Hochbaum and Shmoys \cite{hds87} designed a PTAS for this problem. When the number of machines $m$ is fixed, Horowitz and Sahni \cite{hs76} developed an FPTAS.

On the other hand, if $e_{i,k} \in \{0, 1 \}$  for all time intervals, i.e. at any time the machine is either processing a primary job or a routine job but not both,  then our problems reduce to the problems of parallel machine scheduling with availability constraint where jobs can be resumed on after being interrupted: $P_m| r-a |C_{max}$ and $P_m| r-a |\sum C_j$.  The problem $P_m| r-a |C_{max}$ is NP-hard and
% for both makespan and total completion time.
approximation  algorithms are developed
%for total completion time \cite{ll93} and makespan
in  \cite{k98} and \cite{l91}.
The problem $P_m| r-a |\sum C_j$ is NP-hard when $m=2$ and one machine becomes unavailable after some finite time \cite{ll93}. In the same paper Lee and Liman show that the SPT rule with some modifications leads to a tight relative error of $\tfrac{1}{2}$ for $P_2, NC_{zz} \mid \mid \sum C_j$ where one machine is continuously available and the other machine becomes unavailable after some finite time.

%there is only one if the jobs are not allowed to migrate from one machine to another; and when the jobs are resumable, Leung and Pinedo showed that preemptive SPT rule is optimal when the number of available machines does not go down by 2 during the time period of $p_{max}$, which is the largest processing time of all jobs. (??? missing citation. migration is confusing.)

%In our model, there are multiple machines, the processing capacity of each machine can change between 0 and 1 from interval to interval and the capacity change can happen while a job is being processed. The goal is to find a schedule to minimize total completion time. So far there are no results about the problems studied in this paper. Note that  if the processing capacity is always 1, that is, all machines have full availability, our problem becomes the classical parallel machine scheduling problem $P_m \mid \mid \sum C_j$ which can be solved using SPT (Shortest Processing Time First) rule.
%On the other hand, if the processing capacity is either 0 or 1 at any time, i.e.the machine either has full availability or is totally not available for processing the jobs, then our problem reduces to the problem of parallel machine scheduling with unavailability constraint where jobs can be resumed on after being interrupted. %When the jobs are non-resumable,

\subsection{New Contribution}

In this paper, we generalize the shared processing multitasking model proposed by Hall et. al. in \cite{hll16} to parallel machine environment and allow the processing capacity to be different for different routine jobs.

For the objective of makespan, we show that there is no approximation algorithm for the problem $P_m, e_{i,k} \mid \mid C_{max}$ (that is, the sharing ratios are arbitrary for all shared intervals)
unless $P=NP$. Then for a fixed $e_0$, $0 < e_0 < 1$, we analyze the performance of the LS rule,  LS-ECT  (List Scheduling - Earliest Completion Time) rule, LPT  and LPT-ECT (Largest Processing Time - Earliest Completion Time) rule for the problems $P_m,  e_{i,k} \ge e_0 \mid \mid C_{max}$ and $P_m, e_{i\le m_1,k} \ge e_0 \mid \mid C_{max}$. We then develop an approximation scheme for the problem $P_m, e_{i\le m_1,k} \ge e_0 \mid \mid C_{max}$ whose running times is linear time when the number of machines is a constant.

%we study the total completion time minimization problem on parallel machine where the availability of machines varies over time. Specifically, each machine's time line is composed of multiple continuous intervals with the different constant processing capacitys in each interval.
%
For the objective of total completion time, we show that there is no approximation algorithm for the general problem $P_m, e_{i,k} \mid  \mid \sum C_j$ (that is, the sharing ratios are arbitrary for all intervals)
unless $P=NP$. Then for $P_m, e_{i\le m_1,k} \ge e_0 \mid \mid \sum C_j$, that is, the sharing ratios are at least $e_0$ for the intervals on the first $m_1$ ($ 1 \le m_1 \le m-1$) machines but arbitrary for other $(m-m_1)$ machines, we analyze the performance of the SPT rule and SPT-ECT  (Shortest Processing Time  - Earliest Completion Time) rule, and show that SPT-ECT is a $\ceiling{\tfrac{m}{m_1}}\tfrac{1}{ e_0}$-approximation algorithm while SPT can perform arbitrarily bad. We then develop an approximation scheme for the problem $P_m, e_{i\le m-1,k} \ge e_0 \mid \mid \sum C_j$.

\medskip
\medskip

The paper is organized as follows. In Section 2, we study the problems with the objective of makespan minimization. In Section 3, we study the problems with the objective of the total completion time minimization. At last in Section 4 we draw the concluding remarks.

\section{Makespan Minimization}

\subsection{Hardness for Approximation}

In this section, we show that if the processor sharing intervals and the sharing ratios in these intervals  are arbitrary, then  no approximation algorithm exists.

\begin{theorem}\label{theorem-inapprox}
 Let $\mathbb{N}$ be the set of natural numbers. Let $f(n): \mathbb{N}\rightarrow \mathbb{N}$ be an arbitrary function such that $f(n)>1$. There is no  polynomial time $f(n)$-approximation algorithm for $P_m, e_{i,k} \mid  \mid C_{\max}$  even if $m=2$ unless $P=NP$.
\end{theorem}

\begin{proof}
We prove the inapproximability by reducing from the partition problem.

\noindent {Partition Problem: } Given a set of positive integers $\{a_1,\cdots, a_n\}$ where  $A = \sum_{i=1}^n a_i$ is even, find if the set can be  partitioned into two sets with equal sum $\frac{1}{2}A$.

Given an instance of partition problem, we can reduce it  to an instance of our   scheduling problem  as follows:
There are two machines, $n$ primary jobs, for each job $j$, $p_j = a_j$. There are two routine jobs, one for each machine, both of which are during the interval $[\tfrac{A
}{2}, \tfrac{A}{2} + f(n)\cdot A]$ with the sharing ratio of $e= \tfrac{1}{ f(n) \cdot A}$.

We will show that if there is an $f(n)$-approximation algorithm for the scheduling problem,  then   the algorithm  returns a schedule  with makespan at most $\tfrac{A}{2} \cdot f(n)$ for the constructed  scheduling problem instance if and only if there is a partition to the given partition instance.

Apparently, the makespan of any schedule of the primary jobs is at least $\tfrac{A}{2}$. Also if a job can't finish at $\tfrac{A}{2}$, it will then takes at least $f(n)A$ additional time  due to the processing sharing.    Thus the makespan of any schedule is either exactly $\tfrac{A}{2}$, or at least $\tfrac{A}{2}+f(n)A = \tfrac{A}{2}(1 +
2f(n))$. Therefore, if the approximation algorithm returns a schedule whose makespan is at most $\tfrac{A}{2} \cdot f(n)$, then the makespan of the schedule must be exactly $\tfrac{A}{2}$. This implies that there is a partition to the integers.

On the other hand, if there is a partition, i.e., the numbers can be partitioned into two sets, each    having a sum exactly $\tfrac{A}{2}$, then the optimal schedule  will schedule the corresponding jobs in each set to a machine and its makespan  is exactly $\tfrac{A}{2}$. Then the $f(n)$-approximation algorithm
%If an algorithm is an $f(n)$-approximation algorithm, it
must return a schedule whose makespan is at most  $\tfrac{A}{2} \cdot f(n)$.

The above analysis shows that the integer partition problem has a solution if and only if the algorithm  returns a schedule with completion time at most $\tfrac{A}{2}f(n)$ for the corresponding scheduling instance.  Since the partition problem is NP-hard, unless P=NP, there is no such approximation algorithm.
\end{proof}

Given the inapproximability result from Theorem~\ref{theorem-inapprox}, from now on, we will focus on the problems  $P_m, e_{i,k} \ge e_0 \mid \mid C_{max}$ and $P_m, e_{i\le m-1,k} \ge e_0 \mid \mid C_{max}$.

%consider the cases where there are   some restrictions on the sharing ratio of  some machines. Let $e_0$ be a constant such that $e_0\in (0,1]$, and $m_1$ be an integer such that $1\le m_1 \le m$, we consider the case that $e_{i,j} \ge e_0$ for $1 \le  i \le m_1$. In other words, the sharing ratio is at least $e_0$ for the intervals on the first $m_1$ machines, we denote this problem by  $Pm | \text{share, } e_{i,k} \ge e_0 \text{ for } i \le m_1 |  C_{\max}$. One special case is $e_0 =1 $, that is, no processor sharing for   $m_1$ machines, in this case, the notation can be simplified as $Pm | \text{share, } e_{i,k}=1 \text{ for } i \le m_1 |  C_{\max}$. Another special case is $m_1 = m$, the notation can be simplified as  $Pm | \text{share, } e_{i,k}\ge e_0 |  C_{\max}$.

% First, if we only have $k < m$ machines that have sharing intervals, we denote the problem to be $Pm| {share}_k|  C_{\max}$.
% Second, let $e_0$ be a constant such that $e_0\in (0,1)$. If for all sharing  intervals on all machines we have,  $e_{i,k} \ge e_0$, we denote the problem to be $Pm | \text{share,} e_{i,k}\ge e_0 |  C_{\max}$; if    $e_{i,k} \ge e_0$ holds for  $m_1$ machine only, we denote the problem to be $Pm | \text{share,}e_ {i,k} \ge e_0 \text{i \le m_1} |  C_{\max}$.

\subsection{Approximation Algorithms}
\subsubsection{Preliminary}
First, we consider the problem  $P_m,  e_{i,k} \ge e_0 \mid \mid C_{max}$ where the sharing ratio is bounded below by a constant $e_0$ for all intervals, i.e. $e_{i,k} \ge e_0$ for all machines  $M_i$, $1 \le i \le m$.

We have the following observation: Let $I$ be an instance for $ P_m,  e_{i,k} \ge e_0 \mid \mid C_{\max}$, and $I'$ be the corresponding instance for $ P_m \mid\mid C_{\max}$. Let $S'$ be a schedule for $I'$ and $S$ be the corresponding schedule for $I$. Then the completion time of a job $j$ in $S$ is at most $\tfrac{1}{e_0}$ times that in $S'$. It is easy to see that the optimal makespan for $I$ must be greater than or equal to that of  $I'$, thus we can get the following observation.

 \begin{observation} \label{1-over-c-approx}
 An  $\alpha$-approximation algorithm for $ P_m \mid\mid  C_{\max}$ is an $\tfrac{\alpha}{e_0}$-approximation algorithm for $P_m,  e_{i,k} \ge e_0 \mid \mid C_{\max}$.
\end{observation}

Observation~\ref{1-over-c-approx} and the existing literature  \cite{g69,hds87, hs76} imply the following results for $ P_m,  e_{i,k} \ge e_0 \mid \mid C_{\max}$:
 \begin{enumerate}

    \item LS rule is a $ (2-\tfrac{1}{m})\tfrac{1}{ e_0}$-approximation.
    \item LPT rule is a $(\tfrac {4}{3}-\tfrac{1}{3m})\tfrac{1}{ e_0}$-approximation.
    \item There is a $\tfrac {1+\epsilon}{e_0}$ approximation whose running time is $O(n^{2k}\ceiling{\log_2\tfrac{1}{\epsilon}})$, where $k = \ceiling{\log_{1 + \epsilon} \tfrac{1}{\epsilon}}$.
    \item There is a  $\tfrac {1+\epsilon}{e_0}$ approximation whose running time is $O(n(n^2/\epsilon)^{m-1})$.
 \end{enumerate}

 It turns out the above bound for LS rule is quite loose. In the following, we will give a tighter bound.

 \subsubsection{List Scheduling (LS) rule}

 For an instance of $ P_m,  e_{i,k} \ge e_0 \mid \mid C_{\max}$, let $C_{max}(LS)$ and $C_{max}^*$ be the makespan of an arbitrary list schedule and the optimal schedule,  respectively. Then we have the following theorem.

 \begin{theorem} \label{LS-approx}
 For $ P_m,  e_{i,k} \ge e_0 \mid \mid C_{\max}$,     $\tfrac{C_{\max}(LS)}{C_{max}^*} \le  ( 1 + \tfrac{1}{e_0 })$ .
\end{theorem}
\begin{proof}
 Without loss of generality, we assume $C_{\max}(LS) = C_j$ for some job $j$. Assume that job $j$ starts at time $t$ in the list schedule. This means there is no idle time before $t$ and we must have $ C_{max}^* > t $, and
 $C_{\max}(LS)=C_j  \le  t + \tfrac{p_j}{e_0} \le C_{max}^* + C_{max}^*\cdot \tfrac{1}{e_0} = C_{max}^* ( 1 + \tfrac{1}{e_0 })$.
\end{proof}

\smallskip{ \noindent \bf Complexity of List Scheduling}  Given a list, the list scheduling can be implemented as follows. For a partial schedule, we maintain the completion time of the last job $f_i$, $1 \le i \le m$,  on all machines using a min-heap.
The next job $j$ in the list will be assigned to the machine with minimum $f_i$, which can be found and updated in $O( k + \log m)$ time, where $k$ is the number of intervals during which the job $j$ is scheduled. The total running time to assign all $n$ jobs will be $O(n \log m + \tilde{n})$. If we use LPT rule, then we need additional time to sort the jobs and the total time would be $O(n\log n + n \log m + \tilde{n}) = O (n \log n + \tilde {n})$.

 \subsubsection{Modified List Scheduling - LS-ECT and LPT-ECT}
 If we relax  the constraint on sharing ratios from $ P_m,  e_{i,k} \ge e_0 \mid \mid C_{\max}$,  so that only the intervals on the first $m_1$ machines have bounded sharing ratios, i.e., $P_m, e_{i \le m_1,k} \ge e_0 \mid \mid C_{max}$,
 LS rule could perform very badly. Consider the example of  2 machines, the first machine has a bounded sharing interval $(0, \infty)$ with sharing ratio $e_0$, the second machine has a sharing interval $(0,\infty)$ with sharing ratio $x << e_0$, and we have two jobs of length 1. The LS schedule will schedule one job on each machine and the makespan will be $1/x$, but the optimal schedule will schedule both jobs on the first machine and has a makespan of $\tfrac{2}{e_0}$. The approximation ratio of LS is $\tfrac{e_0}{2 x }$, which  will be arbitrarily large as $x$ gets close to 0.

%Let $I$ be the smallest instance for which LIST has the worst performance. Assume the jobs are ordered  1, 2, \ldots, n, then we must have $C_{max}(\text{LIST}) = C_n$. Suppose not, then we can remove job $n$ from the instance, and for the new instance, the performance ratio of LIST will be the same or worse  the makespan of LIST schedule  will not change, but the optimal makespan may be smaller, which contradicts to the fact that $I$ is the smallest instance for which LIST has the worst performance.

% Assume $C_{\max}(List) = C_j$ for some job $j$.
% By the list schedule, when job $j$ is scheduled at time $t$ on some machine $M_i$ during some interval $I_{i,k}$ with the sharing ratio of $e_{i,k}$, all other machines must be busy at time $t$, and thus we have  $ C_{max}^* > t$, and
%  $C_{\max}(List)=C_j  \le  t + \tfrac{p_j}{e_{i,k}} \le C_{max}^* + C_{max}^*\cdot \tfrac{1}{e_{i,k}} = C_{max}^* ( 1 + \tfrac{1}{e_{i,k}})$. If $i>m_1$ and $e_{i,k}$ is arbitrarily close to 0, then the ratio becomes unbounded and close to $-\infty$. Similarly, SPT rule could also give an approximation ratio as worse as the inverse of the worst unbounded sharing ratio.

%For $Pm||C_{max}$, it has been shown the makespan of the schedule  generated by an arbitrary list is at most  $(2-\tfrac{1}{m})$ times that of the minimum makespan.

%List scheduling has been proven a good heuristic for classical scheduling. The jobs are first ordered according to certain priority rule, and then scheduled one by one to the earliest available machine.

To have a better approximation for the more general problem $P_m, e_{i \le m_1,k} \ge e_0 \mid \mid C_{max}$, we want to consider a modified list scheduling for our model: when we schedule the next job $j$ in the list, instead of scheduling it to the  machine so that it can start as early as possible, schedule it to the machine so that it can complete the earliest.  We will call this heuristic LS-ECT, and we call it LPT-ECT if the list is  in LPT order. We use $C_{max}$({\it LS-ECT}) and $C_{max}$({\it  LPT-ECT}) to denote the makespan of the schedule produced by LS-ECT and LPT-ECT respectively.

Before we do further analysis, we give some simple facts and a  claim that we will use frequently later.
 Let $S$ be any schedule without any unnecessary idle time. Let $f_i$, $1 \le i \le m$, denote the completion time of the last job on machine $M_i$ in $S$.
\begin{enumerate}
\item[] {\bf Fact 1:}  There exists at least one $f_i$ such that   $ f_i \le C_{max}^*$.
\item[] {\bf Fact 2:}  Let $S$ be an LS-ECT schedule. If the last job $j$ on machine $M_{i_1}$ starts at a time later than $f_{i_2}$, then rescheduling job $j$ to $M_{i_2}$ does not decrease its completion time.
\end{enumerate}

% \Paragraph{Approximation Ratio}

% Before we give the analysis of the approximation ratio of LS-ECT, we first give the following claim.

\begin{claim}\label{claim:number-of-small-jobs}
Let $S$ be an LS-ECT schedule and $X$ be the number of jobs that complete after $C_{max}^*$ on the first $m_1$ machines, if $ f_i > C_{max}^*$ for all $1 \le i \le m_1$, then $X \le m-1$.
\end{claim}

%we have $ f_i > C_{max}^*$ for all $1 \le i \le m_1$, and for some machine $M_{i'}$, $i' > m_1$, $f_{i'} < C_{max}^*$. We have the following claim.

%that there are at most $(m-1)$ jobs that complete after $C_{max}^*$ on the first $m_1$ machines.
\begin {proof}  Suppose not, then $X \ge m$.
%, that is, there are at least $m$ jobs that finish after $C_{max}^*$ on the first $m_1$ machines.
Among these $X$ jobs, we remove $(m-m_1)$ jobs so that the last job on the first $m_1$ machines are still completed after $C_{max}^*$, and reschedule each of the $(m-m_1)$ jobs to one of the last $m-m_1$ machines. By Fact 2, the completion times of all these jobs do not decrease, so their  completion times are all still larger than $C_{max}^*$,
%we can move at most  $(m-m_1)$ jobs to those remaining machines so that
and thus all machines are busy after $C_{max}^*$, which contradicts Fact 1. This concludes the proof of the claim.
\end{proof}
\begin{figure}
\centering
 \includegraphics[width=0.8\textwidth, height=2.8in]{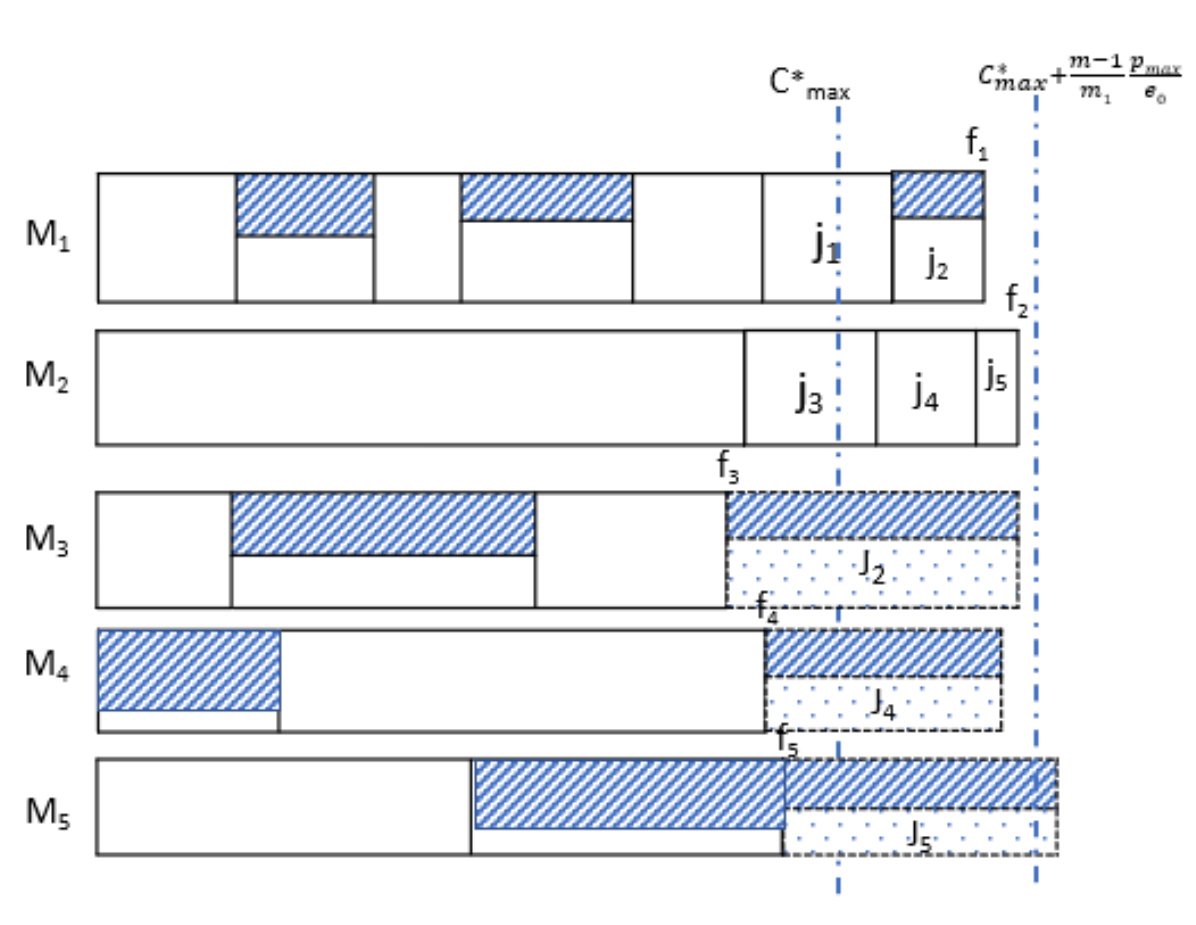}
\caption{An LS-ECT schedule for $P_m, e_{i \le m_1,k} \ge e_0 \mid
\mid C_{max}$ where  $m= 5$, $m_1 =2$, and $f_i \ge C_{max}^*$ for
all $i \le m_1$. It is impossible that 5 or more jobs finish later
after $C_{max}^*$ on the first $m_1=2$ machines.}
\label{fig:aPTAS-greedy-2}
\end{figure}
 Figure~\ref{fig:aPTAS-greedy-2} illustrates the proof of Claim~\ref{claim:number-of-small-jobs} where $m=5$, $m_1=2$,  and $X=5$. Among the 5 jobs that finish after $C_{max}^*$ on the first 2 machines, if we reschedule the jobs  $2, 4, 5$ to the last three machines, respectively,   all the machines are busy after $C_{max}^*$, and this  contradicts Fact 1.

In the following, we consider the performance of LS-ECT for the problem $P_m, e_{i \le m_1,k} \ge e_0 \mid \mid C_{max}$.

\begin{theorem}\label{theorem-LS-ECT-general-m1}
 For $P_m, e_{i \le m_1,k} \ge e_0 \mid \mid C_{max}$, $$\frac{C_{max}(\text{LS-ECT})}{C_{max}^*} \le  \left(1+ \tfrac{1}{e_0} \cdot  \left(\floor{\tfrac{m-1}{m_1}} +1 \right)\right).$$
\end{theorem}

\begin{proof}
Let $I$ be the smallest instance for which LS-ECT has the worst performance. Assume the jobs are ordered  1, 2, \ldots, n, then we must have $C_{max}(\text{LS-ECT}) = C_n$. Suppose not, then we can remove job $n$ to get a new instance. For the new instance, the performance ratio of LS-ECT will be the same or worse because the makespan of LS-ECT schedule  will be the same, but the optimal makespan may be smaller, which contradicts to the fact that $I$ is the smallest instance for which LS-ECT has the worst performance.

Let $S$ be the schedule obtained for $I$ using LS-ECT where $f_i$, $1 \le i \le m$, is the completion time of the last job on machine $M_i$ in the LS-ECT schedule. By Fact 1, there exists at least one $f_i$ such that   $f_i \le C_{max}^*$.

% As in  Theorem~\ref{theorem-LS-ECT-big-m1}, we consider  the smallest instance $I$ for which LS-ECT has the worst performance. Let $S$ be the schedule, and assume the jobs are ordered  1, 2, \ldots, n,
% then we must have $C_{max}(\text{LS-ECT}) = C_n$.
% Let $f_i$, $1 \le i \le m$, denote the completion time of the last job on machine $M_i$ in the LS-ECT schedule.

If  $ f_i \le C_{max}^*$ for some $i \le m_1$,  then if we reschedule job $n$ to $M_i$, its new completion time will be at most
$f_i + p_n/e_0 \le (1 + \tfrac{1}{e_0}) C_{max}^*$. By Fact 2, this will   not be better than its original completion time.   Thus in this case we have $C_{max}(\text{LS-ECT}) = C_n \le  (1 + \tfrac{1}{e_0}) C_{max}^*$. See Figure~\ref{fig:LS-ECT-m1-lastjob} for an illustration.
  \begin{figure}
  \centering
 \includegraphics[width=0.8\textwidth, height=2.0in]{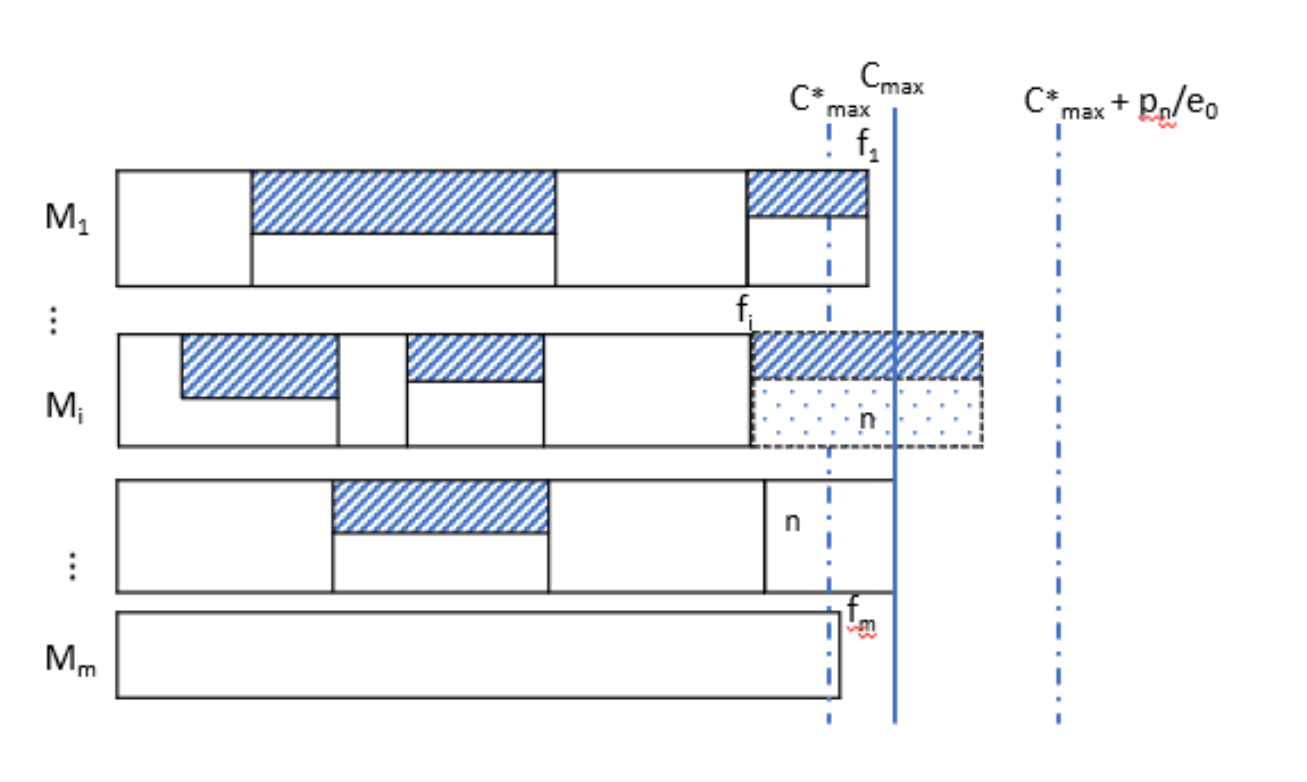}
\caption{An  LS-ECT schedule for $P_m, e_{i \le m_1,k} \ge e_0 \mid \mid C_{max}$  where $f_i \le C_{max}^*$, $i \le m_1$. If job $n$ were rescheduled to $M_i$, it would finish no earlier than its current completion time.} \label{fig:LS-ECT-m1-lastjob}
\end{figure}

 %we get  $C_{max}(\text{LS-ECT}) = C_n \le  (1 + \tfrac{1}{e_0}) C_{max}^*$ as in  Theorem~\ref{theorem-LS-ECT-big-m1}.

Otherwise, we have $ f_i > C_{max}^*$ for all $1 \le i \le m_1$, and for some machine $M_{i'}$, $i' > m_1$, $f_{i'} < C_{max}^*$. By Claim~\ref{claim:number-of-small-jobs}, we know there are at most $X \le m-1$ jobs complete after $C_{max}^*$ on the first $m_1$ machines. By pigeon hole principle, among the first $m_1$ machines, there exists a machine on which there are at most $\floor{X/m_1} \le \floor{(m-1)/m_1}$ jobs finishing after $C_{max}^*$. Without loss of generality, suppose $n$ is not scheduled on this machine, then moving job $n$ to this machine will not decrease its completion time. Let $p_{max} = \max_{1 \le j \le n} p_j$,  then we have
$$C_{max} = C_n \le C_{max}^* +  \floor{\tfrac{m-1}{m_1}}\tfrac{ p_{max}}{e_0} +   \tfrac{p_n}{e_0}  \le \left(1+  \left(\floor{\tfrac{m-1}{m_1}} +1 \right) \cdot \tfrac{1}{e_0} \right) C_{max}^*.$$
\end{proof}

 Theorem~\ref{theorem-LS-ECT-general-m1} implies that if  $m_1 = m$, the  ratio of LS-ECT is $(1+ \tfrac{1}{e_0})$ which agrees with the ratio of LS rule from Theorem~\ref{LS-approx}. If $m_1 = m-1$,
the approximation from    Theorem~\ref{theorem-LS-ECT-general-m1}  becomes $(1 + \tfrac{2}{e_0})$, which is worse than the case $m_1 = m$. In the following we show that in this case, the approximation ratio is actually still bounded by $(1 + \tfrac{1}{e_0})$.

\begin{theorem}\label{theorem-LS-ECT-big-m1}
For $P_m, e_{i \le m-1,k} \ge e_0 \mid \mid C_{max}$,
$$\frac{C_{max}\text{(LS-ECT)}}{C_{max}^*} \le 1 + \tfrac{1}{e_0 }.$$
\end{theorem}
\begin{proof}
Let $I$ be the smallest instance for which LS-ECT has the worst performance. Assume the jobs are ordered  1, 2, \ldots, n, then we must have $C_{max}(\text{LS-ECT}) = C_n$.

Let $S$ be the schedule obtained for $I$ using LS-ECT where $f_i$, $1 \le i \le m$, is the completion time of the last job on machine $M_i$ in $S$. By Fact 1, there exists at least one $f_i$ such that $f_i \le C_{max}^*$.

If $i \le m-1$, then as in the proof of Theorem~\ref{theorem-LS-ECT-general-m1},
 %if we reschedule job $n$ to $M_i$, its completion time will be at most
%$f_i + p_n/e_0 \le (1 + \tfrac{1}{e_0}) C_{max}^*$. By Fact 2, this will be not better than its original completion time  we use LS-ECT.   Thus in this case
we have $C_{max}(\text{LS-ECT}) = C_n \le  (1 + \tfrac{1}{e_0}) C_{max}^*$.

 Otherwise, there is no $i$ such that  $i \le m-1$  and $f_i \le C_{max}^*$. This implies (1) on each of the first $(m-1)$ machines, at least one job finishes after $C_{max}^*$, (2) $f_m < C_{max}^*$, and (3) job $n$ must be  scheduled on one of the first $(m-1)$ machines.  By Claim~\ref{claim:number-of-small-jobs}, there are at most $(m-1)$ jobs that complete after $C_{max}^*$ on the first $m-1$ machines. Combining with (1), there is exactly one job that completes after $C_{max}^*$  for each $M_i$, $i\le m-1$.
 %
 %Suppose this is not the case, i.e., there are two jobs that finish after $C_{max}^*$ on some machine $M_i$, $i \le m-1$. By Fact 2, if we reschedule the last job (whose completion time is $f_i > C_{max}^*$) on $M_i$ to $M_m$, its completion time will not be decreased and is still larger than $C_{max}^*$. In this new schedule, all machines are still busy after $C_{max}^*$. However, this is impossible by Fact 1.
Therefore, $n$ must be the only job that finishes after $C_{max}^*$ on the machine it is scheduled; and it starts at or before $C_{max}^*$. So we have $C_{max}(\text{LS-ECT}) = C_n \le  C_{max}^*  + p_n/e_0 \le (1 + \tfrac{1}{e_0}) C_{max}^*$.
\end{proof}

A natural question follows from %Theorem~\ref{theorem-LS-ECT-general-m1} and
Theorem~\ref{theorem-LS-ECT-big-m1} is if the performance of LS-ECT is still $(1 + \tfrac{1}{e_0})$ when $m_1 = m-2$. We show by an example that this is unfortunately not the case. We have $m=3$ and $m_1 = 1$. The first machine has a sharing ratio of ${e_0}$ during interval $[x+2, \infty)$, the other 2 machines have  sharing ratio $\tfrac{e_0}{3x}$ during interval $[x, \infty)$, and all other intervals have sharing ratio 1; there are 5 jobs whose processing times are x, 1, 1, x, x. If we use  this list for LS-ECT, the first machine has three large jobs, and the other two machines each has one small job, and the makespan  is $(x+2 + \tfrac{2x-2}{e_0})$; however in the optimal schedule,  one large job and the two small jobs are scheduled on the first machine, and the other two large jobs are scheduled on the second and third machines, respectively. The  makespan of the optimal schedule is $(x+2)$.  The performance  ratio of LS-ECT approaches to $(1 + \tfrac{2}{e_0})$  as $x$ increases.

Note that if $e_0 = 1$, the approximation ratio obtained from Theorem~\ref{theorem-LS-ECT-big-m1} becomes 2 which is very close to the   ratio  $(2-1/m)$  of LS rule for classical model  where
all machines have full capacity.

We also want to point out that  for  LPT-ECT, we can get slightly better ratios than those for LS-ECT in Theorem~\ref{theorem-LS-ECT-general-m1} and \ref{theorem-LS-ECT-big-m1}.   %$(1 + \tfrac {m}{n e_0})$.
Based on the fact
 $C_{max}^* \ge \tfrac{\sum p_j} {m}$, and the fact for LPT that $ p_n \le \tfrac{\sum p_j}{ n} $, we have $ p_n / e_0  \le   \tfrac {m}{n e_0}   C_{max}^*$. Then we can get the following result for  LPT-ECT.
 \begin{corollary}
For $P_m, e_{i \le m_1,k} \ge e_0 \mid \mid C_{max}$, $$\frac{C_{max}(\text{LPT-ECT})}{C_{max}^*} \le  \left(1+ \frac{1}{e_0} \cdot  \left(\floor{\frac{m-1}{m_1}} + \frac {m}{n} \right)\right),$$
and for $P_m, e_{i \le m-1,k} \ge e_0 \mid \mid C_{max}$, $$\frac{C_{max}(\text{LPT-ECT})}{C_{max}^*} \le 1 + \frac {m}{n e_0}.$$
\end{corollary}

%\subsubsection{Data Structure}
\paragraph{Time Complexity}\label{ECT-time}

%Paragraph{Data Structure}

%Paragraph{ Complexity of ECT scheduling}
Compared with LS rule, LS-ECT takes more time  for each job because we need to compare its completion times on all machines and schedule it to the machine so that it completes at the earliest time.

To implement LS-ECT, considering a partial schedule where  jobs $1$, $2$, $\ldots$, $j-1$ have been scheduled, we maintain a pair $(P_i, f_i)$ for each machine $M_i (1 \le i \le m)$, where $P_i$ is the total processing time of the jobs that have been assigned to $M_i$, and $f_i$ is the completion time of the last job on $M_i$. Moreover, for each interval on machine $M_i$,  $I_{i,1} = (0, t_{i,1}]$, $I_{i,2} = (t_{i,1}, t_{i,2}]$, $\ldots$, we maintain a quadruple $(t_{i,k-1}, t_{i,k}, e_{i,k}, A_i(t_{i,k}))$ where $e_{i,k}$  is the sharing ratio of the interval $(t_{i,k-1}, t_{i,k}])$ and $A_i(t_{i,k})$ is the total amount of the jobs that can be scheduled before $t_{i,k}$.
For convenience of scheduling, we pre-calculate $A_i(t_{i,k})$ for every $
t_{i,k}$ on machine $M_i$.  Note that $A_i(t_{i,k}) = A_i(t_{i,k-1}) + e_{i,k}(t_{i,k} - t_{i,k-1})$. Assuming the intervals are given  in sorted order, the calculation can be done in $O(k_i)$ time for machine $M_i$ where $k_i$ is the number of shared intervals on $M_i$, and in $O(\tilde {n})$ for all $m$ machines.

%We use a triple $(t_{i,k-1}, t_{i,k}, e_{i,k})$ to represent an interval $[t_{i,k-1}, t_{i,k}])$ on machine $M_i$ with the sharing ratio $e_{i,k}$, $0 < e_{i,k} \le 1$. For any time point $t$, we let $A_i(t)$ denote the total amount of processing time of jobs that can be processed during $(0,t]$ on machine $M_i$, that is, if $t_{i,k} < t \le t_{i,k+1}$, $A_i(t) = ( \sum_{u \le k} (t_{i,u}-t_{i,u-1})e_{i,u}) + (t-t_{i,k}) e_{i,k+1}$.

To assign a job $j$ to this partial schedule, for each machine $M_i$, we can use binary search on the continuous intervals $I_{i,1}$, $I_{i,2}$, $\ldots$, to find the interval $(t_{i,k-1}, t_{i,k}]$ such that $A_i(t_{i,k-1}) <  P_i + p_j \le A_i(t_{i,k}$). Then job $j$'s completion time on machine $M_i$ can be calculated as $t = t_{i,k-1} + \tfrac{P_i + p_j - A_i(t_{i,k-1})}{e_{i,k}}$.
%where $e_{i,k}$ is the sharing ratio of the interval $[t_{i,k-1}, t_{i,k})$.
After all machines are considered, we assign job $j$ to the machine so it  completes the earliest.
%Suppose there are $k_i$ intervals on $M_i$.
In total, it takes $O((\sum_{i=1}^{m} \log k_i)+m) $ time to assign a  job to the machine so it completes the earliest. So the overall time for scheduling $n$ jobs using LS-ECT is $O(\tilde {n} + n(m  + \sum_{i=1}^m \log k_i))$, and   the total time for  LPT-ECT is  $O(n\log n + \tilde {n} + n(m  + \sum_{i=1}^m \log k_i))$.

\subsubsection{Comparison of List Scheduling and Modified List Scheduling}

While LS gives the approximation ratio of $( 1 + \tfrac{1}{e_0 })$ for problem $P_m, e_{i,k} \ge e_0 \mid \mid C_{max}$, it can perform badly
for problem $P_m, e_{i \le m_1,k} \ge e_0 \mid \mid C_{max}$. In comparison, the approximation ratio of  LS-ECT for $P_m, e_{i \le m_1,k} \ge e_0 \mid \mid C_{max}$ is  $\left(1+ \tfrac{1}{e_0} \cdot  \left(\floor{\tfrac{m-1}{m_1}} +1 \right)\right)$. And when $m_1=m-1$, LS-ECT gives the approximation ratio of $( 1 + \tfrac{1}{e_0 })$. Similar conclusions hold  for LPT rule and LPT-ECT rule.

Now we consider the  performance of LPT and LPT-ECT when the number of jobs is small. When $n \le 2m$, LPT is optimal for $P_m \mid \mid C_{\max}$, however, this does not hold any more for $P_m, e_{i,k} \ge e_0 \mid \mid C_{max}$ even if $n=2$. Consider the example of two machines and two jobs of length 1. Suppose the first machine doesn't have processor sharing, while the second machine has a sharing interval $[0, \infty)$ with $e_{2,1} = e_0 < \tfrac{1}{2}$. The optimal schedule has both jobs on the first machine with the makespan of 2, but LPT rule schedules one job on each machine and thus has an approximation ratio of $\tfrac{1}{2e_0}$. Although LPT-ECT can find the optimal schedule when  there are only two jobs, it is not optimal anymore when the number of jobs $n \ge 3$.
Consider three jobs of length 3, 2, 2  and two machines. There is a single sharing interval $I_{2,1}=[0, \infty)$, and $e_{2,1} = 3/4$. The LPT-ECT will schedule one shorter job on  $M_2$, and the other two jobs on $M_1$ with the makespan  $5$ while the optimal schedule schedules the longest job on $M_2$ and the 2 shorter jobs on $M_1$, and the makespan is 4.

\subsection{Approximation Scheme}

%In this section, we develop two different approximation schemes that have incomparable running time. When the number of the machines is constant, both run in linear time.
%\subsection{Approximation Scheme Using LS-ECT}

In this section, we develop an %polynomial time
approximation scheme for $P_m, e_{i \le m_1,k} \ge e_0 \mid \mid C_{max}$. The idea is to partition the jobs into two groups, one for large jobs and the other for small jobs. Then we  schedule the large jobs using enumeration; and schedule the small jobs using LS-ECT.  Let $d$ be the number of large jobs which determines the error ratio of the output schedule and will be specified later. Our algorithm is formally presented as follows.
%The number of  large jobs $d$ determines   the error ratio of the output schedule.

\medskip
\noindent {\bf Algorithm 1}

\smallskip
 {\noindent \bf Input:}
     \begin{itemize}
         \item Parameters $m$, $n$, $m_1$,  and $e_0$ %the lower bound $e_0$ of sharing ratios  on the first $m_1$ machines
          \item The intervals   $(0, t_{i,1}]$, $(t_{i,1}, t_{i,2}]$, $\ldots$ on machine $M_i$, $1 \le i \le m$,
          and their sharing ratios $e_{i,1}$, $e_{i,2}$, $\ldots$, respectively
        \item The jobs' processing time $p_j$, $1 \le j \le n$
        \item Integer parameter $d$ that determines the accuracy of approximation
     \end{itemize}
%$n$ jobs; an integer $d$; the lower bound of sharing ratio $e_0$, the intervals on each machine in sorted order, this includes those with full capacity as well, the total number intervals is  $O(\tilde{n}) = O(\sum_{1 \le i \le m} k_i)$.

  {\noindent \bf Output:}  a schedule of the $n$ jobs

  \medskip
   {\noindent \bf Steps:}
\begin{enumerate}
  % \item[1.] Sort the jobs in LPT order, compute $A_i(t_{i,k})$ for all $t_{i,k}$
    \item Find the $d$ largest jobs
%   \item Obtain all possible assignments of $d$ large jobs on $m$ machines by enumeration
    \item For each possible assignment of the large jobs
    \item[] \indent Schedule the remaining jobs using LS-ECT
    \item Return the schedule $S$ obtained from previous step that has the minimum makespan
\end{enumerate}

\noindent For ease of analysis, we first analyze the performance of Algorithm 1 for $m_1 = m$, i.e. $P_m, e_{i,k} \ge e_0 \mid \mid C_{max}$, then the general case $m_1 \le m$, i.e., $P_m, e_{i \le m_1,k} \ge e_0 \mid \mid C_{max}$. Let $S^*$ be an optimal schedule and  $C_{\max}^*$ be the makespan of $S^*$.

\begin{lemma}\label{lemma-algorithm1-error} For an instance of  $P_m, e_{i,k} \ge e_0 \mid \mid C_{max}$, Algorithm 1   returns a schedule  with the makespan at most $(1+\tfrac{m}{d \cdot e_0}) \cdot C_{\max}^*$.
\end{lemma}

\begin{proof}
 Let $p_{d}$ be the processing time of the $d$-th largest job. Then  we must have $C^*_{\max} \ge \tfrac{d \cdot p_{d}}{m}$. Thus, $p_d \le \tfrac{m}{d} C^*_{\max}$.

Let $S'$ be the schedule obtained from step 2 of Algorithm 1 that
has the same large job assignment as $S^*$. Let $j$ be the job such
that $C_j(S') = C_{max}(S')$. Without loss of generality, we can
assume $C_{max}(S') > C_{\max}^*$. Then $j$ must be a small job. By
Fact 1, there is at least one machine such that the last job on this
machine finishes at or before $C_{\max}^*$ in any schedule. Let
$M_i$ be such a machine in $S'$. Then job $j$ must be scheduled   on
a machine other than $M_i$ in $S'$. By Fact 2, if job $j$ is
rescheduled to $M_i$ in $S'$, then its new completion time, at most
$C^*_{\max} + \frac{p_j } {e_0}$, would not be less than its
original completion time $C_{\max}(S')$ (see
Figure~\ref{fig:aPTAS-greedy-1}). Therefore, we have
\begin{eqnarray*}
C_{\max}(S') & \le  & C^*_{\max} + \frac{p_j } { e_0} \\
& \le   & C^*_{\max} + \frac{1}{ e_0} p_d \\
& \le & C^*_{\max} + \frac{1}{ e_0} \cdot \tfrac{m}{d} C^*_{\max} \\
&\le& \left(1+\frac{m}{d \cdot e_0 }\right) \cdot C^*_{\max}.
\end{eqnarray*}
 \begin{figure}
 \includegraphics[width=\textwidth, height=2.0in]{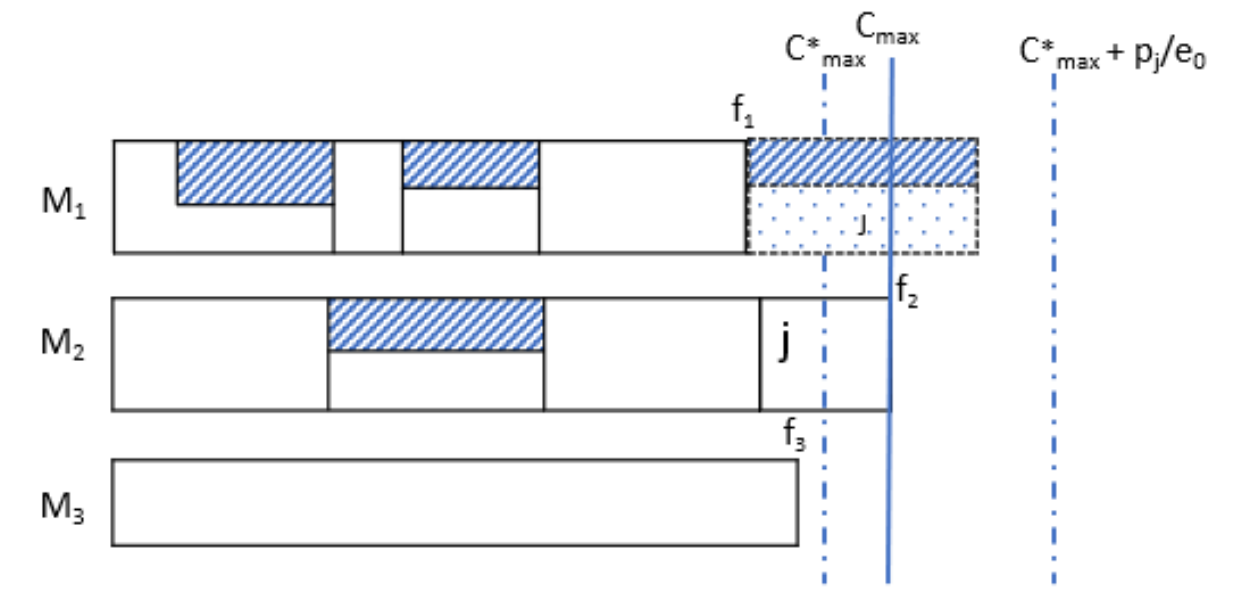}
\caption{Illustration of schedule $S'$ obtained from Algorithm 1:
$C_{max}(S') = C_j$, if job $j$ were scheduled on $M_1$, it would
finish no earlier than its current completion time.}
\label{fig:aPTAS-greedy-1}
\end{figure}

Since Algorithm~1 returns a schedule $S$ with minimum makespan, the above bound is also an upper bound of $C_{\max}(S)$.
%
%Now we consider the running time of Algorithm~1.  We  first sort the primary jobs in $O( n \log n)$ time. Then it is easy to find the first $d$ largest jobs.
%there are at most $m^d$ ways to  assign these large jobs to $m$ machines.
%For each of these partial schedules, scheduling the remaining jobs in LPT-ECT takes at most $O((\tilde{n}+m) (n-d)) = O((m + \tilde{n}) n)$ time. Thus the total time is $O(n \log n + m^d (m + \tilde{n})n)$.
\end{proof}

\begin{lemma}\label{lemma-algorithm1-time}
%Let $1 \le d \le n$  be an integer.
Algorithms 1 can be implemented in $O( \tilde{n} + n+ (m^{d} (n-d) (m+ \sum_{i=1}^{m} \log k_i)) )$ time.
\end{lemma}
\begin{proof}
%In step 1, we sort the $n$ jobs in LPT order, which can be done in $O( n \log n)$ time. We can compute all $A_i(t_{i,k})$ in $O(\tilde{n})$ time.
In step 1, %$d$ largest jobs can be found in $O(n)$ time.
we first find the $d$-th largest job  using linear selection algorithm, and then extract the $d$ largest jobs  in $O(n)$ time.  In step 2, there are at most $m^d$ ways to assign these large jobs to $m$ machines. For each large job assignment, the algorithm schedules the remaining small jobs using LS-ECT rule. As we described in the Section \ref{ECT-time} for the time complexity analysis, with  $A_i(t_{i,k})$ pre-calculated in $O(\tilde{n})$ time, it takes $O((\sum_{i=1}^{m} \log k_i)+m) $ time to assign a  job to the machine so it completes the earliest.
So  for each large job assignment, the overall time for scheduling $(n-d)$ small jobs using LS-ECT in step 2 is $O((n-d)(m + \sum_{i=1}^m \log k_i))$. Adding all the time,  we get the total running time of Algorithm 1 $ O( \tilde{n} + n +
m^d(n-d) (m + \sum_{i=1}^{m} \log k_i)) $
%= O( \tilde{n} + m^d(n-d) (m + \sum_{i=1}^{m} \log k_i))$.
%
%
%It is easy to see  $U_{i,t}=\sum_{v=1}^{t-1} |I_{i,v}| s_{I_{i,v}}$, which
%characterizes that amount of work to finish in the last $t-1$
%intervals at this machine.
%
%to $O(m^{d}n(1+ \sum_{i=1}^m \log k_i ) $ ( $ m + $ ???  ).
\end{proof}

Given  an instance of $P_m, e_{i,k} \ge e_0 \mid \mid C_{max}$, and a  real number   $ \epsilon \in (0,1)$,  if  we select $d=\ceiling{\tfrac{m} { \epsilon e_0}}$ and apply algorithm 1, then by Lemma~\ref{lemma-algorithm1-error}, we get a schedule $S$ whose makespan is at most  $(1 + \epsilon)C^*_{\max}$. Combining Lemma~\ref{lemma-algorithm1-error} and Lemma~\ref{lemma-algorithm1-time}, we get the following theorem.

\begin{theorem}\label{theorem-algorithm1-PTAS}
For any given instance of  $P_m, e_{i,k} \ge e_0 \mid \mid C_{max}$ and an error parameter $\epsilon$,  $0< \epsilon < 1 $,  Algorithm 1 can  return a schedule  with makespan at most $(1+ \epsilon)C_{max}^*$  in $O( \tilde{n}+n +m^{{m}/({\epsilon \cdot e_0})} (n- \tfrac{m} {\epsilon \cdot e_0}) (m+ \sum_{i=1}^{m} \log k_i) )$ time.
\end{theorem}

Next we analyze the performance of Algorithm 1 for the more general problem $P_m, e_{i \le m_1,k} \ge e_0 \mid \mid C_{max}$.
We will show that by choosing $d$ appropriately, we can still get a $(1 + \epsilon)$ approximation.

\begin{theorem}\label{approximation-scheme-thm-linear-sharing-bound-on-one-machine} For any given instance of  $P_m, e_{i \le m_1,k} \ge e_0 \mid \mid C_{max}$ and a parameter $d$, Algorithm~1 returns a schedule with makespan at most $(1+\frac{m(m+m_1 -1)}{d \cdot e_0 \cdot m_1})C_{max}^*$ in $O( \tilde{n} + n+
m^d (n-d) (m+ \sum_{i=1}^m \log k_i) )$ time.  In particular, if $d=\ceiling {\tfrac{m(m+m_1-1)}
{\epsilon \cdot e_0 \cdot m_1}}$, it is a $(1+\epsilon)$-approximation.
\end{theorem}

\begin{proof}
As in the proof of Lemma~\ref{lemma-algorithm1-error}, we consider the schedule $S'$ from step~2 of Algorithm~1 that has the same assignment of large jobs as the optimal schedule $S^*$. Let $j$ be the job such that $C_j(S') = C_{\max}(S')$. Without loss of generality, we assume $C_{\max}(S')>C_{\max}(S^*)$, then $j$ must be a small job.
Let $p_{d}$ be the processing time of the $d$-th largest job. Then we must have $C^*_{\max} \ge \tfrac{d \cdot p_{d}}{m}$, that is, $p_d \le \tfrac{m}{d} C^*_{\max}$.
By Fact 1, there is at least one machine in $S'$ where the last job finishes at or before $C_{\max}^*$. If there exists one such machine $M_i$ with $i \le m_1$, we can use similar argument as that of Lemma~\ref{lemma-algorithm1-error},  to show that
\begin{eqnarray*}
C_{\max}(S') = C_j & \le  & C^*_{\max} + \frac{p_j } { e_0} \\
&\le& \left(1+\frac{m}{d \cdot e_0 }\right) \cdot C^*_{\max}
.
\end{eqnarray*}

Otherwise, for all machines $M_i$, $1 \le i \le m_1$, the last job finishes after $C_{\max}^*$. By Claim~\ref{claim:number-of-small-jobs}, there are at most $m-1$ jobs that finish after $C_{\max}^*$ on these $m_1$ machines and there must exist one machine where at most $\floor{\frac{m-1}{m_1}}$ jobs finish after $C_{\max}^*$. Similarly, by Fact 2, moving job $j$ to this machine does not decrease its completion time.
 %there are at most $(m-1)$ jobs that finish after  $C_{\max}^*$  on all of the $m_1$ machines in the schedule $S'$.
%Suppose not, then there are at least $m$ jobs that finish after $C_{\max}^*$.
%   we use LPT-ECT scheduled these jobs, if we move one of these  jobs to a   machine whose last job finish before $C_{\max}^*$,  then its new completion time must be greater than its current completion time, thus greater than $C_{\max}^*$.  there are  at most $(m-m_1)$ machine where the last job finish at or before $C_{\max}^*$, but there are at least $m$ jobs that finish after $C_{\max}^*$ on the first $m_1$ machines,  we can reschedule some of these jobs to those machines, so that the last job of all machines completes after $C_{\max}^*$,
  %
%Also there exists one of the first $m_1$ machines on which there are at most $\floor{(m-1)/m_1}$ jobs completes after $C_{\max}^*$, let it be $M_k$, $k \le m_1$. Then the completion time of last job on $M_k$ is at most $C_{\max}^* + \tfrac{m-1}/{m_1}  \tfrac{p_d }{ e_0}$.
%
%\smallskip
%, we bound  $C_{\max}(S')$. Similarly as before,   let $p_{d}$ be the processing time of the $d$-th largest job. Then  we must have $C^*_{\max} \ge \tfrac{d}{m} p_{d}$. Thus, $p_d \le \tfrac{m}{d} C^*_{\max}$.
%In addition, in $S'$, if we move the  last job  on any  machine  to $M_k$,  its new completion time will be increased to at most $C_{\max}^* + (\tfrac{m-1}{m_1} +1 )\tfrac{p_d }{ e_0}$.
Therefore,
\begin{eqnarray*}
C_{\max}(S') = C_j(S') & \le  & C^*_{\max} +\left(\floor{\frac{m-1}{m_1}} +1 \right) \frac {p_d} { e_0} \\
& \le  & C^*_{\max} + \frac{m+m_1-1}{ m_1}  \cdot \frac{m}{d} C^*_{\max} \cdot \frac{1}{e_0}\\
%& \le & C^*_{\max} + \frac{(1+e_0)}{ e_0^2} \cdot \tfrac{m}{d} C^*_{\max} \\
&\le& \left(1+\frac{m(m+m_1-1)}{d \cdot e_0 \cdot m_1 }\right) \cdot C^*_{\max}.
\end{eqnarray*}

 Algorithm~1 returns a schedule $S$ that is at least as good as $S'$, so the above bound is also an upper bound of $C_{\max}(S)$.
Let $\epsilon$ be a real number in $(0,1)$,  if   we select $d=\ceiling{\tfrac {m(m+m_1-1)} { \epsilon \cdot e_0 \cdot m_1}}$, then $C_{\max}(S) \le (1 + \epsilon)C^*_{\max}$.

Finally, the analysis of running time remains the same as in Lemma~\ref{lemma-algorithm1-time}.
%Now we consider the running time of Algorithm~1.  We  first sort the primary jobs in $O( n \log n)$ time. Then it is easy to find the first $d$ largest jobs.
%there are at most $m^d$ ways to  assign these large jobs to $m$ machines.
%For each of these partial schedules, scheduling the remaining jobs in LPT-ECT takes at most $O(\tilde{n} n)$ time. ???
\end{proof}

\section{Total Completion Time Minimization}

\subsection{Hardness of Approximation}
%{Complexity of the Problem}

It is known that the classical problem $P_m || \sum C_j$  can be solved using SPT and it  becomes inapproximable when the machines have  unavailable periods, that is, the sharing ratio is in $\{0, 1\}$. In this section, we show that, the problem
$P_m, e_{i,k} \mid \mid \sum C_j$
does not  admit any approximation algorithm even if there are only two machines and the sharing ratio is always positive. %$e_{i,k} >0$. %, which has the arbitrary sharing ratios on all machines, % is  not only NP-hard, but also

\begin{theorem}\label{thm:sum-cj-inapproximation}Let $f(n): \mathbb{N}\rightarrow \mathbb{N}$ be an arbitrary function such that $f(n)>1$.
There is no  polynomial time $f(n)$-approximation algorithm for $P_m, e_{i,k}\mid \mid \sum C_j$ unless $P=NP$.
\end{theorem}

\begin{proof}
We reduce from the  partitioned problem.
In the Partition Problem, we are given a set of positive integers $\{a_1, a_2, \cdots, a_n\}$, where $A = \sum_{i = 1}^{n} a_i$. The problem is ``can the set be partitioned into two subsets with equal sum $\tfrac{A}{2}$?''
 We construct an instance of $P_m, e_{i,k} \mid \mid \sum C_j$ as follows:
There are $n$ jobs, 1, 2, \ldots, n, job $j$ has processing time $p_j = a_j$.
There are 2 machines, each having the intervals $(0, \tfrac{A}{2}]$ and $(\tfrac{A}{2}, +\infty)$ with the sharing ratios of $e_{i,1} =1 $ and $e_{i,2} = \tfrac{1}{nf(n)\cdot A}$, $i = 1, 2$, respectively.

It is easy to see that there is a partition of the set if and only if there is a schedule where all jobs could finish before $\tfrac{A}{2}$, i.e.  the total completion time of all jobs is at most $\tfrac{nA}{2}$.
We can show further, the latter problem can be answered if there is a $f(n)$-approximation algorithm.

Suppose there is a schedule in which the total completion time of all jobs is at most $\tfrac{nA}{2}$, then a $f(n)$-approximation algorithm would return a schedule with the total completion time at most $\tfrac{nA}{2}\cdot f(n)$. This implies that all jobs must finish before $\tfrac{A}{2}$ because if a job finishes after $\tfrac{A}{2}$, its completion time will be at least $ \tfrac{A}{2} + \tfrac{1}{e_{i,2}} = \tfrac{A}{2}+nA \cdot f(n)$, and thus the total completion time is greater than $\tfrac{n A}{2} \cdot f(n)$. %So the returned   schedule by the algorithm must have total completion time at most $\tfrac{nS}{2}$.
Hence, there exists a schedule whose total completion time at most $\tfrac{nA}{2}$ if and only if the $f(n)$-approximation algorithm returns a schedule such that all jobs finish before $\tfrac{A}{2}$. Consequently, we can solve the partition problem, which is impossible unless  $P=NP$.
\end{proof}

Given the inapproximability result from Theorem~\ref{thm:sum-cj-inapproximation}, from now on, we will focus on the problems such that the sharing ratios on some machines are greater than or equal to  a  constant  $e_0$.

\subsection{Approximation Algorithms}

In this section, we study the problem when there exist one or more machines such that the sharing ratios for all intervals on these machines have a constant lower bound, that is, $P_m, e_{i\le m_1,k} \ge e_0 \mid \mid \sum C_j$. %One can easily prove the NP-hardness of this problem by reducing the problem of $P||C_{max}$ to it. Hence, our focus is to develop the approximation algorithms for this problem.
We first analyze the performance of SPT
%(Shortest Processing Time )
and its variant SPT-ECT
% (Shortest processing Time - Earliest Completion Time)
for our problem,
%both of which generate an optimal schedule  for the classical problem $P_m || \sum C_j$.
and then we develop a PTAS.

\subsubsection{SPT and SPT-ECT}

%SPT rule schedules the next shortest job in the list to the earliest available machine. If two machines $M_i$ and $M_j$ are both available at the same time, we assume the machine with smaller index will be used.
It is well known that for the classical problem $P_m || \sum C_j$, SPT generates an optimal schedule where the jobs complete in SPT order.
%,  SPT generates the same optimal schedule where the jobs complete in SPT order.
Its variant, SPT-ECT rule, which schedules the next shortest job to the machine so it completes the earliest, generates the same optimal schedule as SPT.

Now we consider SPT and SPT-ECT rules for our problem $P_m,  e_{i\le m_1,k} \ge e_0 \mid \mid \sum C_j$. First of all,   SPT and SPT-ECT may generate difference schedules. Consider two machines where the first machine has sharing ratio 1 during the interval  $(0, 1]$ and $\tfrac{1}{2}$ during  $(1, \infty)$, the second machine has sharing ratio $1$ all the time. There are 3 jobs whose processing times are 1, 2, 2.  SPT will schedule one job  of length 2 on $M_2$, and two other jobs on $M_1$. The total completion time is $1 + 5 + 2 = 8$. SPT-ECT may schedule the first job on $M_1$ and the other two jobs on $M_2$, the total completion time is $1 + 2 + 4 = 7$. Moreover, SPT and SPT-ECT don't dominate each other, i.e. for some cases SPT-ECT generates better schedule (see the above example), while for some other cases, SPT generates better schedule. For the above example, if we add one more job of processing time $3$, SPT will generate a better schedule which schedules jobs with length 1 and 2 on $M_1$ and the other two jobs on $M_2$. The total completion time is $ 1 + 5 + 2 + 5 = 13$. The SPT-ECT, on the other hand, may schedule the jobs with length 1 and 3 on $M_1$ and the other two jobs on $M_2$, the total completion time is then $1 + 2 + 4 + 7 = 14$.

%For example, there are four jobs $J_1$, $J_2$, $J_3$, $J_4$ with the processing times of $1$, $1+\epsilon$, $1+\epsilon$, $1+2\epsilon$ and there are two machines with one processing sharing interval $[1, \infty)$ of sharing ratio $\tfrac{1}{2}$ on the first machines. SPT will schedule jobs $J_1$, $J_3$ on $M_1$ and jobs $J_2$, $J_4$ on $M_2$ with the total completion time of $7+6\epsilon$. SPT-ECT may schedule jobs $J_1$, $J_4$ on $M_1$ and jobs $J_2$, $J_3$ on $M_2$ with the total completion time of $7+7\epsilon$, which is worse than the schedule produced by SPT.

Next we show that the approximation ratio of SPT rule for our problem is unbounded.  Consider two machines where sharing ratio $e_{1,1} =1$ during $[0, \infty)$ and  $e_{2,1} = \tfrac{1}{\alpha}$ during  $[0, \infty)$.
Given an example of two jobs with processing time 1 and 1,
SPT rule schedules two jobs one on each machine with the total completion time of $1+\tfrac{1}{\alpha}$ while in the optimal schedule,  both jobs are on $M_1$ with the total completion time of $2$, which is optimal. The approximation ratio for SPT for this instance is $\tfrac{1+({1}/{\alpha})}{2}$. The ratio approaches infinity when $\alpha$ is close to 0.

Finally we prove that the approximation ratio of SPT-ECT rule for our problem is bounded. For convenience, we first prove the following claim before we give the approximation ratio of SPT-ECT.
\begin{claim}\label{m1-verses-m-SPT}
 Given a set of $n$ jobs and machines with % full availability (that is, the
 sharing ratio always 1, the minimum total completion time of the jobs on $m_1$ identical machines is at most $\ceiling{\tfrac{m}{m_1}}$ times that on $m$, $m \ge m_1$ identical machines. \end{claim}
\begin{proof}
Suppose the jobs are indexed by SPT order. We first consider the case that $m$ is a multiple of $m_1$, i.e.  $\tfrac{m}{m_1} = c$ for some integer $c$.

Let  $S^*_{m_1}$ be the optimal schedule for $m_1$ machines which is generated by SPT.
Then in $S_{m_1}^*$,  the indices of jobs scheduled on $M_i$ will be $i$, $m_1 + i$, $2m_1 + i$, $\ldots$
Let $j = \lambda m_1 + i$ be a job on $M_i$, then  its completion time will be
$$C_j (S^*_{m_1}) = \sum_{ 0 \le k \le \lambda} (p_{ k m_1 + i}).$$

Let  $S^*_{m}$ be the optimal schedule for $m$ machines which is generated by SPT.  It is easy to see that the jobs scheduled on $M_i$, $1 \le i \le m_1$,  in $S_{m_1}^*$ are now scheduled in SPT order on $\tfrac{m}{m_1} = c$ machines:  $M_i$, $M_{m_1 + i}$, $M_{2m_1 + i}$, $\ldots$, $M_{(c - 1)m_1+i}$ in $S_m^*$.   For  job $j = \lambda m_1 + i$, $ 1 \le i \le m_1$, a lower bound for $C_j(S^*_{m})$ can be obtained by assuming that these $c$ machines are processing the jobs $ ( k m_1 + i )$ only, $0 \le k \le \lambda$, between time 0 and $C_j(S^*_{m})$,
$$C_j (S^*_{m}) \ge \tfrac{ \sum_{ k =0} ^{\lambda} {p_{ k m_1 + i}}}{c} = \tfrac{1}{c}    C_j (S^*_{m_1})   = \tfrac{m_1}{m}  C_j (S^*_{m_1}), $$
which means

$$\sum_{j=1}^n C_j(S_{m_1}^*)  \le
\tfrac{m}{m_1} \sum C_j (S^*_{m}).$$

For the case that $m$ is not a multiple of $m_1$,  let $m'> m$ be the smallest multiple of $m_1$. Using the above argument,   the minimum total completion time on $m_1$ machines is at most $\tfrac{m'}{m_1}= \ceiling{\tfrac{m}{m_1}} $ times the minimum total completion time on $m'$ machines; the latter is  a lower bound on the total completion time on $m$ machines,
This completes the proof.
\end{proof}

\noindent Now we give the approximation ratio of SPT-ECT for our problem.

\begin{theorem}\label{m1-e0-verses-m-SPT}
For $P_m, e_{i\ge m_1,k} \ge e_0 || \sum C_j$, SPT-ECT is $\ceiling{\tfrac{m}{m_1}}\tfrac{1}{ e_0}$-approximation with the running time of $O(n\log n + \tilde {n} + n(m  + \sum_{i=1}^m \log k_i))$.
\end{theorem}

\begin{proof}
%\noindent Now we are ready to prove the Theorem~\ref{m1-e0-verses-m-SPT}.
%
%\noindent {\bf Proof of Theorem ~\ref{m1-e0-verses-m-SPT}:}
Let $S_m^*$ be the optimal schedule for the jobs on all $m$ machines,
and let $\widetilde{S}_m^*$ %, and $\tilde{S}_{m_1}^*$
be the optimal schedule  of the jobs on $m$ machines without processor sharing. It is obvious that  $ \sum C_j (\widetilde{S}_m^*) \le \sum C_j (S_m^*)$.
Let $\widetilde{S}_{m_1}$ be the schedule obtained by applying SPT-ECT to the $m_1$ machines without processor sharing, which is the same as the schedule obtained by applying SPT.
By Theorem~\ref{m1-verses-m-SPT},
$$\sum_{j=1}^n C_j(\widetilde{S}_{m_1})  \le \
\ceiling{\tfrac{m}{m_1} } \sum  C_j (\widetilde{S}_m^*).$$

Let $S$ be the SPT-ECT schedule of the jobs on all $m$ machines. Let $S_{m_1}$ be the schedule obtained by applying SPT-ECT to the first $m_1$ machines only. Then the total completion time of $S$ is at most that of $S_{m_1}$, which is at most $\tfrac{1}{e_0}$ times that of  $\widetilde{S}_{m_1}$
Thus, we have
 $$ \sum C_j (S) \le \sum C_j (S_{m_1}) \le \tfrac{1}{e_0} \sum C_j(\widetilde{S}_{m_1}) \le \ceiling{\tfrac{m}{m_1}}\tfrac{1}{ e_0}  \sum C_j (\widetilde{S}_m^*) \le \ceiling{\tfrac{m}{m_1}}\tfrac{1}{ e_0}  \sum C_j (S_m^*).$$

 With the same implementation of LPT-ECT for the makespan minimization problem, the total time would be $O(n\log n + \tilde {n} + n(m  + \sum_{i=1}^m \log k_i))$.
%\QED
\end{proof}

\subsection{Approximation Scheme}\label{PTAS1}

In this section, we develop a PTAS for our problem when  the sharing ratios  on $(m-1)$ machines  have a lower bound $e_0$, i.e.,  $P_m,    e_{i\le m-1, k} \ge e_0 | | \sum C_j$.
For convenience, we introduce the following two notations that will be used in this section.
\begin{enumerate}
    \item[]  $P_i(S)$: the total processing time of the jobs assigned to machine $M_i$ in $S$. % use $t_i(S)$ to denote the completion time of the last job on $M_i$,
  \item[] $\sigma_i(S)$: the total completion time of the jobs scheduled to $M_i$ in $S$.
\end{enumerate}
The idea of our algorithm is to schedule the jobs one by one in SPT  order; for each job $j$ to be scheduled, we enumerate all the possible assignments of job $j$ to all machines $M_i$, $1 \le i \le m$, and then we prune  the set of schedules so that no two schedules are  ``similar''. % and we only  keep   a limited number of schedules.
Two schedules  $S_1$ and $S_2$  are ``similar'' with respect to a give parameter $\delta$  if for every $1 \le i \le m$, $P_i(S_1)$ and $P_i(S_2)$ are both in an interval $[(1+\delta)^x, (1 + \delta)^{x+1})$ for some integer $x$, and $\sigma_i(S_1)$ and $\sigma_i(S_2)$ are both in an interval $[(1+\delta)^y, (1 + \delta)^{y+1})$ for some integer $y$. We use $S_1 \overset{\delta}{ \approx}  S_2$ to denote that $S_1$ and $S_2$ are ``similar'' with respect to $\delta$. Our algorithm is formally presented as follows.

\medskip

\noindent{\bf Algorithm2}

\smallskip

\noindent {\bf Input}:
\begin{itemize}
\item $e_0$, $\epsilon$
 \item The intervals   $(0, t_{i,1}]$, $(t_{i,1}, t_{i,2}]$, $\ldots$ on machine $M_i$, $1 \le i \le m$,
          and their sharing ratios $e_{i,1}$, $e_{i,2}$, $\ldots$, respectively.
%\item $m$ machines, each machine $M_i$ is associated with $k_i$ sharing ratios $e_{i, k}$, $k \le k_i$.
For $1 \le i \le m-1$ and $ 1 \le k \le k_i$,  $e_{i, k} \ge e_0$.    %, the first $(m-1)$ machines with processing sharing intervals such that $e_{i,j} \ge e_0$ and the $m$-th machine with the processing sharing intervals of arbitrary sharing ratios
\item $n$ jobs with the processing times, $p_1$, $\cdots$, $p_n$ %in the increasing order
\end{itemize}

\noindent{\bf Output:} A schedule  $S$ whose total completion time is at most $(1 + \epsilon)$ times the optimal.

\noindent{\bf Steps:}

\begin{enumerate}
    \item Reindex the jobs in SPT order
    \item  Let  $\delta = \tfrac{\epsilon \cdot e_0}{6 n}$
    \item Let $U_0 = \emptyset$
    \item For $j=1, \ldots, n$, compute $U_j$ which is a set of  schedules of the first $j$ jobs:
     \item[]  \begin{enumerate}
        \item $U_j = \emptyset$
        \item for each schedule $S_{j-1} \in U_{j-1}$
        \item[] \hspace{0.2in} for $i = 1 \ldots m$
        \item[] \hspace{0.4in} add job j to the end of $M_i$ in $S_{j-1}$, let the schedule be $S_j$
        \item[] \hspace{0.4in} $U_{j} = U_j \cup \{S_j \}$
        \item prune $U_{j}$ by repeating the following until $U_{j}$ can't be reduced
        %\item[] Let $\delta < \tfrac{\epsilon \cdot e_0}{ 6n}$.
        \item[] \hspace{0.2in} if there are two schedules $S_1$ and $S_2$  in $U_j$ such that $S_1  \overset{\delta}{ \approx } S_2$
        \item[] \hspace{0.5in} if $P_m(S_1) \le P_m(S_2)$,
        $U_j = U_j \setminus \{S_2\}$
        \item[] \hspace{0.5in} else   $U_j = U_j \setminus \{S_1\}$
    \end{enumerate}
    \item Return the schedule $S \in U_n$ that minimizes $\sum_{ i = 1}^{m} \sigma_i$
\end{enumerate}

\begin{theorem}\label{approximation-scheme-thm2} Algorithm2  is a $(1+\epsilon)$-approximation scheme for
 $P_m,  e_{i\le m-1,k} \ge e_0 | | \sum C_j$, and it runs in time $O(n \log n+ n (m+\tilde{n}) (\tfrac{36}{\epsilon^2 \cdot e_0^2}(\log P) (\log \tfrac{nP}{e_0})n^2)^m)$, where $P = \sum p_j$.
\end{theorem}

\begin{proof}
Let $S^*$ be the optimal schedule. %By lemma~\ref{SPT-on-each-machine-claim}, the jobs on each machine are scheduled in SPT order.
We use $S^*_j$ to denote the partial schedule of the first $j$ jobs in $S^*$.
We first prove by induction the following claim:
%properties of the partial schedules in $U_j$, $0 \le j \le U_j$.

\medskip
%{\noindent \bf Claim:}
For each job $1 \le j \le n$,   there is a partial schedule $S_j \in U_j$ such that
\begin{enumerate}[topsep=0pt,itemsep=-1ex,partopsep=1ex,parsep=1ex]
    \item[] Property (1): $  P_m(S_j)  \le P_m(S^*_j)$,
    \item[]Property (2)   $P_i(S_j)  \le (1 + \delta)^j P_i(S^*_j)$ for $1 \le i \le m-1$, and
    \item[] Property (3)
    % $t_i(S_j)  \le (1 + \tfrac{ \delta}{e_0} ) t_i(S^*_j)$ for $1 \le i \le m$, and
    %\item[(4)]
     $\sigma_i(S_j) \le (1 + \delta)^j (1+ \tfrac{ 2 n \delta}{e_0}) \sigma_i(S^*_j)$ for $1 \le i \le m$.
\end{enumerate}
 \medskip
  It is trivial for $j=1$.
Assume the hypothesis is true for $j$, so we have a schedule $S_j \in U_j$ with properties (1)-(3).  Consider the schedule of job $j+1$ in $S^*$.

{\noindent \bf Case 1.} In $S^*$, job $j+1$ is scheduled on $M_m$. Then  $P_m(S^*_{j+1}) = P_m(S^*_{j}) + p_{j+1} $. Let $S_{j+1}$ be the schedule obtained from $S_j$ by scheduling job $(j+1)$ on $M_m$.   Then we have
\setcounter{equation}{3}
\begin{eqnarray}
P_m(S_{j+1}) &  = & P_m(S_{j}) + p_{j+1} \nonumber \\ & \le &  P_m(S^*_{j}) + p_{j+1}  \hspace{0.3in} \text { by Property (1)} \nonumber\\
& = &  P_m(S^*_{j+1}) \enspace.
\end{eqnarray}
For $ 1 \le i \le m-1$, $S_j$ and $S_{j+1}$ are the same on $M_i$, so are schedules $S^*_{j}$ and $S^*_{j+1}$. Thus, we have
\begin{equation} P_i (S_{j+1}) = P_i (S_j) \le (1 + \delta)^j P_i (S^*_{j})  =   (1 + \delta) ^j P_i (S^*_{j+1})\enspace, \end{equation}
and for $1 \le i \le m-1$,
\begin{equation} \sigma_i (S_{j+1}) = \sigma_i (S_j) \le (1 + \delta)^j(1 + \tfrac{2 n \delta}{e_0}) \sigma_i (S^*_{j}) = (1 +   \delta )^j (1 + \tfrac{2 n \delta}{e_0}) \sigma_i (S^*_{j+1}).\end{equation}
For $i = m$, note that (4) implies    $C_{j+1}(S_{j+1}) \le C_{j+1}(S^*_{j+1})$,
\begin{eqnarray}
 \sigma_m(S_{j+1}) & = &  \sigma_m(S_{j}) + C_{j+1}(S_{j+1}) \nonumber \\
 & \le &  (1 + \delta)^j(1 + \tfrac{2n \delta}{e_0})\sigma_m(S^*_{j}) + C_{j+1}(S_{j+1} )  \hspace{0.2in} \text { by Property (3)}  \nonumber \\
 & \le &  (1 + \delta)^j(1 + \tfrac{2n \delta}{e_0})\sigma_m(S^*_{j}) + C_{j+1}(S^*_{j+1} )  \hspace{0.2in}    \nonumber\\
 & \le & (1 + \delta)^j(1 + \tfrac{2n \delta}{e_0})\left(\sigma_m(S^*_{j}) + C_{j+1}(S^*_{j+1})\right) \nonumber \\
 & =  & (1 + \delta)^j(1 + \tfrac{2 n \delta}{e_0})\sigma_m(S^*_{j+1}) \enspace.
\end{eqnarray}
If $S_{j+1}$ is not pruned, i.e. $S_{j+1} \in U_{j+1}$,  inequalities (4)-(7) mean $S_{j+1}$ is the schedule in $U_{j+1}$ with Properties (1)-(3). Otherwise, there must exist another schedule  $S'_{j+1} \in U_{j+1}$ such that $S'_{j+1} \overset{\delta}{\approx} S_{j+1}$ which implies
 $$P_m(S'_{j+1}) \le  P_m (S_{j+1}),$$
 $$P_i(S'_{j+1}) \le (1 + \delta) P_i (S_{j+1}) \text{ for } 1 \le i \le m-1 , \text{ and }$$
 $$\sigma_i(S'_{j+1}) \le (1 + \delta) \sigma_i (S_{j+1})  \text{ for } 1 \le i \le m \enspace. $$

Combining with the above inequalities (4)-(7), we get:
   $$P_m(S'_{j+1}) \le  P_m (S^*_{j+1}),$$
 $$P_i(S'_{j+1}) \le (1 + \delta)^{j+1} P_i (S^*_{j+1}) \text{ for } 1 \le i \le m-1 , \text{ and }$$
 $$\sigma_i(S'_{j+1}) \le (1 + \delta)^{j+1}(1+\tfrac{2 n \delta}{e_0} \sigma_i (S_{j+1}))  \text{ for } 1 \le i \le m \enspace. $$
 Thus,  $S'_{j+1}$ is the schedule in $U_{j+1}$ with Properties (1) - (3) in this case.

{ \noindent \bf Case 2.} In $S^*$, $j+1$ is scheduled on  machine $M_k$, $1 \le k \le m-1$. Let $S_{j+1}$ be the schedule obtained from $S_j$ by scheduling job $(j+1)$ on $M_k$.
Then for $M_k$ in $S_{j+1}$, we have
\begin{eqnarray*}
P_k(S_{j+1}) & = & P_k(S_{j}) + p_{j+1}\\
             & \le &  (1 + \delta)^{j} P_k (S^*_{j}) + p_{j+1} \\
 & \le &  (1 + \delta)^{j} ( P_k (S^*_{j}) + p_{j+1} ) \\
 & = &  (1 + \delta)^{j} P_k (S^*_{j+1}).
\end{eqnarray*}
Thus, ${ P_k(S_{j+1})- P_k(S^*_{j+1})}\le  (( 1+ \delta)^j -1 ) P_k(S^*_{j+1})$. Since $(j+1)$ is scheduled on $M_k$ in both $S_{j+1}$ and $S^*_{j+1}$, the difference of the completion time of $j+1$ in these schedules will be
$$C_{j+1}(S_{j+1}) - C_{j+1}(S^*_{j+1}) \le \tfrac{P_k(S_{j+1})- P_k(S^*_{j+1})} {e_0} \le  \tfrac {(( 1+ \delta)^j -1 ) P_k(S^*_{j+1})}{e_0}.$$
Obviously $P_k(S^*_{j+1}) \le C_{j+1}(S^*_{j+1})$; and  we can show that
$$( 1+ \delta)^j -1  \le ( 1+ \delta)^n -1 =
\sum_{i=1}^n \binom{n}{i} {\delta}^i
\le  n\delta ( \sum_{i=1}^n{\frac{1}{ i!}})  \le 2n\delta,$$
thus $C_{j+1}(S_{j+1}) - C_{j+1}(S^*_{j+1}) \le  \tfrac{ 2n \delta}{e_0} C_{j+1}(S^*_{j+1})$. Therefore,
\begin{eqnarray*}
\sigma_k(S_{j+1}) & = & \sigma_k(S_{j}) + C_{j+1}(S_{j+1})\\
             & \le &  (1 + \delta)^{j}\left(1 + \tfrac{2n\delta}{e_0} \right)\sigma_k (S^*_{j}) + C_{j+1}(S_{j+1}) \\
%  & \le &  (1 + \delta)^{j}(1 + \tfrac{2 n \delta}{e_0} )\sigma_k (S^*_{j}) + C_{j+1}(S^*_{j+1}) +  \tfrac{ P_k(S_{j+1})- P_k(S^*_{j+1})}{e_0} \\
%  & < &  (1 + \delta)^{j}(1 + \tfrac{2 n \delta}{e_0} )\sigma_k (S^*_{j}) + C_{j+1}(S^*_{j+1}) +  \tfrac{ (( 1+ \delta)^j -1 ) P_k(S^*_{j+1})}{e_0} \\
% & < &  (1 + \delta)^{j}(1 + \tfrac{2 n \delta}{e_0} )\sigma_k (S^*_{j}) + C_{j+1}(S^*_{j+1}) +  \tfrac{ (( 1+ \delta)^j -1 ) C^*_{j+1}(S^*_{j+1})}{e_0}\\
% & < &  (1 + \delta)^{j}(1 + \tfrac{ 2 n \delta}{e_0} )\sigma_k (S^*_{j}) + (1  +  \tfrac{ ( 1+ \delta)^j -1 } {e_0})C_{j+1}(S^*_{j+1}) \\
& < &  (1 + \delta)^{j}\left(1 + \tfrac{ 2 n \delta}{e_0} \right)\sigma_k (S^*_{j}) + \left(1  +  \tfrac{ 2 n  \delta } {e_0}\right)C_{j+1}(S^*_{j+1}) \\
& < &  (1 + \delta)^{j}\left(1 + \tfrac{ 2 n \delta}{e_0} \right)(\sigma_k (S^*_{j}) + C_{j+1}(S^*_{j+1})) \\
& = &  (1 + \delta)^{j}\left(1 + \tfrac{ 2 n \delta}{e_0} \right)\sigma_k (S^*_{j+1}).
 \end{eqnarray*}
Since for other machines, the schedules $S_j$ and $S_{j+1}$ are same and so are the schedules $S^*_{j}$ and $S^*_{j+1}$,  we have
$$P_m(S_{j+1} ) < P_m(S^*_{j+1})$$
$$P_i(S_{j+1} ) <(1  + \delta) ^j P_i(S^*_{j+1}), \text { for } i \neq m, k,  \text { and } $$
$$\sigma_i(S_{j+1}) <(1  + \delta) ^j (1 + 2 n \delta) \sigma_i(S^*_{j+1} ), \text { for } 1 \le i \le m$$
Same as Case1,  we can show the properties (1) - (3)  hold no matter  $S_{j+1} \in U_{j+1}$ or not.

\bigskip
Thus at the end of the algorithm, after all $n$ jobs have been processed,
\begin{eqnarray*}
\sum_{j=1}^n C_j (S) & = & \sum_{i=1}^m \sigma_i (S)  \\
& \le & (1 + \delta)^n \left(1 + \tfrac{2 n \delta} {e_0}\right)  \sum_{i=1}^m \sigma_i (S^*) \\
& = & (1 + \delta)^n \left(1 + \tfrac{2 n \delta} {e_0}\right)  \sum_{j=1}^n C_j (S^*) \\
& \le  & (1 + 2 n \delta) \left(1 + \tfrac{2 n \delta} {e_0}\right)  \sum_{j=1}^n C_j (S^*) \\
& \le  &  \left (1 + \tfrac{2 n \delta} {e_0}\right)^2  \sum_{j=1}^n C_j (S^*)
\end{eqnarray*}
By step 2 of Algorithm2, $\delta = \tfrac{ \epsilon e_0}{ 6 n} $, then
$$ \sum_{j=1}^n C_j (S) \le \left (1 + \tfrac{2 n \delta} {e_0}\right)^2  \sum_{j=1}^n C_j (S^*)
 \le  \left(1 + \tfrac{ \epsilon}{3}\right)^2  \sum_{j=1}^n C_j (S^*) \le  \left(1 + \epsilon\right)\sum_{j=1}^n C_j (S^*).
$$

Now we consider the running time. Let $P = \sum p_j$, then  after $U_j$ is pruned,   for any schedule $S \in U_j$ and for any machine $M_i$, $1 \le i \le m$, $P_i(S)$ can take at most $\log_{1+\delta} P$ values; $\sigma_i(S)$ can take at most  $\log_{1+\delta}(\tfrac{ nP}{e_0})$ values.  Thus after pruning, there are at most $(\log_{1+\delta} P\log_{1+\delta}\tfrac{ nP}{e_0})^m$ schedules in $U_j$. In each iteration when job $j$ is added, at most $\tilde{n}$ sharing intervals are considered for all machines, so the total time for Step 4 of Algorithm2 will be $O(n(m+\tilde{n})(\log_{1+\delta} P\log_{1+\delta}\tfrac{ nP}{e_0})^m )$.
Plugging in $ \delta = \tfrac{\epsilon \cdot e_0}{6 n}$ and adding the sorting time to get the jobs in SPT order, the total time is $O(n \log n+ n (m+\tilde{n}) (\tfrac{36}{\epsilon^2 \cdot e_0^2}(\log P) (\log \tfrac{nP}{e_0})n^2)^m)$.
\end{proof}

\section{Conclusions}

In this paper we studied the problem of multitasking scheduling with shared processing in the parallel machine environment and with different fractions of machine capacity assigned to different routine jobs. The objectives are minimizing makespan and minimizing the total completion time. %This research is motivated by the work of Hall et. al. in \cite{hll16}, but a generalization from their single machine shared processing multitasking scheduling problem with equal sharing ratio.
For both criteria, we proved the inapproximability for the problem with arbitrary sharing ratios for all machines. Then we focused on the problem where some machines have a positive constant lower bound for the sharing ratios. %For this problem we analyze the performance of the popular scheduling rules including list scheduling, modified list scheduling (LS-ECT), LPT, and LPT-ECT
For makespan minimization problem, we analyzed the approximation ratios of some LS based algorithms and developed an approximation scheme that runs in linear time when the number of machines is a constant. %Most importantly, the running time of these two approximation schemes are not comparable.
For total completion time minimization, we studied the performance of SPT and SPT-ECT and developed an approximation scheme. Our work extends the existing research to analyze scheduling models with multitasking features and develop efficient algorithms, which is an important research direction due to their prevalent applications in business and service industry.

Our research leaves one unsolved case for the total completion time minimization problem: is there an approximation scheme when $m>2$ and more than one machines have arbitrary sharing ratios? For the future work, it is also interesting to study other performance criteria including maximum tardiness, the total number of tardy jobs and other machine environments such as uniform machines, flowshop, etc. Moreover, in our work we assume that all the routine jobs have been predetermined and have fixed release times and duration as well as fixed processing capacity. In some applications, routine jobs may have relaxed time windows to be processed  and flexible processing capacity that could be changed depending on where in the window the routine jobs are processed. For  this scenario, one needs to consider the schedules for both primary jobs and routine jobs simultaneously.

\bibliography{bibliography}
%\end{document}

\end{document}